\documentclass[10pt]{article}
\usepackage[preprint]{tmlr}

\usepackage[utf8]{inputenc}
\usepackage{amsmath, amssymb, amsthm}
\usepackage{booktabs}
\usepackage{longtable}
\usepackage{graphicx}
\usepackage{algorithm}
\usepackage{algpseudocode}
\usepackage{xcolor}
\usepackage{microtype}
\usepackage{hyperref}
\usepackage{url}
\theoremstyle{plain}
\newtheorem{theorem}{Theorem}
\newtheorem{lemma}{Lemma}
\newtheorem{proposition}{Proposition}
\newtheorem{corollary}{Corollary}
\theoremstyle{definition}
\newtheorem{definition}{Definition}

\newcommand{\indep}{\perp\!\!\!\perp}
\newcommand{\E}{\mathbb{E}}
\newcommand{\Var}{\mathrm{Var}}
\newcommand{\Cov}{\mathrm{Cov}}
\newcommand{\given}{\,\vert\,}

\title{GFCM: A Tail-Sensitive Mixed-Type Conditional Independence Test for Causal Discovery}

\author{\name Pavel Averin \email averin.p1@live.unic.ac.cy \\
      \addr Department of Computer Science \\ School of Sciences and Engineering\\University of Nicosia \\ 2417, Nicosia, Cyprus
      \AND
      \name Theodoros Moysiadis \email moysiadis.t@unic.ac.cy \\
      \addr Department of Computer Science \\ School of Sciences and Engineering\\University of Nicosia \\ 2417, Nicosia, Cyprus
      \AND
      \name Ioannis Katakis \email katakis.i@unic.ac.cy \\
      \addr Department of Computer Science \\ School of Sciences and Engineering\\University of Nicosia \\ 2417, Nicosia, Cyprus}

\def\month{MM}  
\def\year{YYYY} 
\def\openreview{\url{https://openreview.net/forum?id=XXXX}} 

\begin{document}
\maketitle

\begin{abstract}
Constraint-based causal discovery like PC and FCI depends on its conditional independence
test. Partial correlation and the Generalised Covariance Measure (GCM) detect only the
conditional covariance of residuals, and therefore miss dependence in the mean's nonlinear
part, the scale, and the tails. Tests that detect more are biased inside PC, or not scalable, or
only continuous, or not aiming at the tails. Our proposal, the Generalised Feature Covariance Measure (GFCM)
makes a test that is valid, sensitive beyond covariance, robust inside PC, and applicable to mixed-type
data. It runs the GCM template on a
configurable set of residual features with conditional mean zero (centered moments and
conditional quantile indicators), pooled in blocks and combined by the Cauchy rule, with a growing-knot spline
nuisance at regression cost. We contribute (i) a centering result making the scale feature
Neyman orthogonal, where the uncentered version is biased; (ii) the orientation asymmetry the
mean-quantile construction creates inside PC, and its fix; (iii) a $\Phi$-faithfulness theory
under which PC with GFCM recovers the CPDAG of the set's detection class; and (iv) a systematic benchmark of CI tests sensitive beyond covariance on synthetic data (continuous and mixed-type) semi-synthetic real data tail injections, and PC discovery on random DAGs. Empirically, under size-corrected
power, GFCM recovers the scale and tail edges the covariance family misses and alone keeps power
at the deep conditioning sets PC issues. It remains calibrated as $n$ grows, whereas FFCI,
the boosted GCM, and the partial copula test do not. The test handles
mixed-type data directly. Inside PC at scale, it achieves the lowest skeleton SHD among the tests
that stay calibrated; the others inflate false edges. Validity rests on an additive nuisance; the tail advantage is shown on simulated and
semi-synthetic data, as no fully real benchmark with both heavy tails and known structure exists.
\end{abstract}

\section{Introduction}
\label{sec:intro}

Constraint-based methods such as PC \citep{algo65} and FCI \citep{algo18_1} learn causal structure by asking, for each pair of variables $X,Y$ and conditioning set $Z$, whether $X$ and $Y$ are independent given $Z$. The quality of the recovered graph therefore depends on the CI test making the correct statistical decision at each such query. 

The CI tests used in practice almost all detect dependence only through the conditional mean or covariance. Partial correlation (Fisher's Z) for continuous data \citep{ci29}, and GCM \citep{ci32} with its variants WGCM \citep{ci50} and PCM \citep{ci52} when nonlinear conditioning is needed, all are in this family. So none of them can see the dependence that lives in the variance or in the tails. Even PCM, the strongest, reaches the conditional mean, but no further.

This blindness is consequential, because the dependence that matters most is often in the tails. In finance, equity correlations that look modest in calm periods spike precisely during crashes \citep{app34}, and systemic risk and spillovers live in joint tail, which measures based on variance and correlation systematically understate \citep{app35,app36}. Heavy tails are the rule rather than the exception in financial and insurance data \citep{app31}. The same pattern holds in climate, where consequential events are extremes whose joint, compound occurrence is missed by analyses of mean dependence \citep{app37}. Crucially, the dependence structure in the tails can differ qualitatively from that in the bulk. In EU carbon markets the network of extreme dependence reorganises relative to that of mean dependence, with peripheral assets becoming central under stress \citep{app30}, and in heavy-tailed structural causal models the relationships among the extremes can differ from those in the centre of the distribution \citep{algo96}. A test that resolves only the conditional mean or covariance therefore recovers the bulk graph and is structurally blind to edges that act through scale or the tails.

This gap has motivated a literature on causal and graphical models built specifically for extremes \citep{algo95,algo96,algo97,algo98,app41}. That work is powerful but specialized: it relies on extreme-value machinery, typically treats continuous data restricted to the tails, in modest dimension, and stands apart from the constraint-based pipelines used in practice. We take a complementary route: instead of a bespoke extremal method, a general conditional independence test that brings location-scale and tail sensitivity into standard constraint-based discovery on mixed-type data.

Tests that look beyond the covariance already exist, and we do not claim that sensitivity as new. Kernel tests such as KCI \citep{ci24} and GKCM \citep{ci60}, nearest-neighbor conditional mutual information \citep{ci10}, and distance correlation \citep{ci55} all detect the dependence above. The spectral GCM (SGCM) \citep{ci67} has a broader detection class and, through characteristic kernel constructions, even reaches non-Euclidean data. The quantile regression partial copula test of \citet{ci90} is the closest predecessor in construction, forming copula residuals by the probability integral transform and testing their dependence. It estimates each conditional distribution by quantile regression over a grid of quantile levels, which is much heavier per query than our reused ridge fits.

The exact kernel and operator tests are impractical for the repeated querying PC demands (KCI is $O(n^3)$, and SGCM rests on a multiplier bootstrap evaluated only up to $n=600$ in its own study), but the fast residual tests are both quick and well calibrated, so speed alone no longer separates a test sensitive beyond covariance from one that relies on it.

The fast tests, though, split on what they give up. RCoT \citep{ci9} and the recent BLITZ \citep{ci100} are well calibrated but not natively mixed-type (RCoT is continuous only; BLITZ handles categoricals only through a jittered integer coding that is not invariant to relabeling; 
the native mixed-type kernel estimators \citep{ci41} that exist are among the slowest), while the Fourier feature CI test FFCI \citep{ci105} is native to mixed-type data but, like the others, untargeted and not robust in the regime PC stresses. On heavy-tailed conditioning sets its size climbs past nominal and keeps rising with sample size, and every fast competitor loses power as the conditioning set grows. None can be aimed at the scale or the tails, so all need far more data to detect dependence that lives purely in conditional skew (Section~\ref{sec:experiments}). Most importantly, none is hardened for the directional and conditioning asymmetries that arise inside PC. The gap this paper fills is therefore a conjunction: a fast, mixed-type test that can target the tails and holds its calibration and its power as both sample size and conditioning depth grow, hardened for constraint-based causal discovery.

We close this gap with the Generalised Feature Covariance Measure (GFCM): a regression-cost, mixed-type, well calibrated conditional independence test that runs correctly inside constraint-based causal discovery algorithms.

\paragraph{Contributions.} The contribution is not detection beyond covariance, which is established; that idea, quantile CI testing \citep{ci90}, and the spectral feature expansion \citep{ci67} all predate us. It is a construction that makes such a test usable inside constraint-based discovery: a correctness result, an identification theorem linking the test to the graph it recovers, two findings on PC integration, and a mixed-type construction.
\begin{enumerate}

\item \textbf{The test.} GFCM (Section~\ref{sec:method}) runs the GCM template at regression cost, on a configurable set of residual features targeting the mean, the variance, and the tails, rather than the single mean residual. It is mixed-type by construction.

\item \textbf{Running inside PC} (Section~\ref{sec:fixes}). A test sensitive beyond covariance can fail inside PC in two ways we identify and repair. \textbf{F1} (Proposition~\ref{prop:f1}): the mean-quantile construction is directionally asymmetric, so for a scale edge one orientation has no signal and PC deletes the true edge; running both orientations repairs it. \textbf{F2}: when PC conditions on a heavy-tailed variable, a spline nuisance with fixed knots loses calibration, its size inflating toward one; rank-transforming the conditioning set restores nominal level at the same cost.

\item \textbf{Theory} (Section~\ref{sec:theory}). We show that the scale residual feature (the absolute residual $|e|$ by default), which detects variance and scale dependence, is biased without centering, and that the centered version satisfies the orthogonality condition our validity theory requires (Lemma~\ref{lem:orth}). We also characterise the set's detection class: which alternatives it provably catches and which it misses (Proposition~\ref{thm:detection}). Validity then follows by verifying the GCM/DML conditions for each feature, binding on a nonparametric nuisance rate. Finally we connect the test to discovery: defining $\Phi$-faithfulness (Definition~\ref{def:phifaith}), we prove that PC with GFCM consistently recovers the CPDAG of the graph whose conditional independence relation is the set's detection class (Theorem~\ref{thm:phicpdag}), so what the set detects and what the algorithm recovers are the same object.

\item \textbf{Empirics} (Section~\ref{sec:experiments}). Across synthetic benchmarks, GFCM is the only test in our panel that keeps both its calibration and its size-corrected power as sample size and conditioning depth grow. On a heavy-tailed null it stays near nominal at the conditioning depths PC queries and approaches nominal with sample size, while the fast mixed-type test FFCI, the boosted GCM, and the partial copula test of \citet{ci90} see their size climb past $0.2$ and keep rising with $n$ (full sweep in Appendix~\ref{app:experiments}, Table~\ref{tab:conv-null}). Under size-corrected power it recovers the location-scale and tail edges the mean and covariance family misses, and at the deep conditioning sets PC issues it is the only test to retain power while the competitors fall to chance. It runs orders of magnitude faster than the kernel and operator tests and natively on mixed-type data. These advantages carry into discovery on random DAGs, where PC+GFCM attains the lowest SHD at scale among tests that control their false positives. They also carry onto real data through two semi-synthetic injections (financial innovations and the Causal Chamber light tunnel), where the covariance family fails to detect the scale and tail edges GFCM recovers.

\end{enumerate}

The remainder of the paper is organized as follows. Section~\ref{sec:prelim} sets up the preliminaries, DAGs, CPDAGs, the PC algorithm, and faithfulness, and the GCM template we build on. Section~\ref{sec:method}
presents GFCM, its feature set, and the two fixes that let it run inside PC.
Section~\ref{sec:theory} establishes validity, the detection class, and the $\Phi$-faithfulness
identification result. Section~\ref{sec:experiments} reports the simulation studies, the real-data
injections, and discovery inside PC. Section~\ref{sec:related}
reviews related work, and Section~\ref{sec:discussion} discusses scope, limitations, and
conclusions. Proofs are deferred to the appendix.

\section{Preliminaries}
\label{sec:prelim}

\paragraph{Graphs, $d$-separation, and Markov equivalence.} A directed acyclic graph (DAG) \citep{def2}
$G$ over a variable set $V$ encodes conditional independences through $d$-separation \citep{def4}:
disjoint sets $X,Y$ are $d$-separated by $S$ in $G$ if $S$ blocks every path between them. A
distribution $P$ is Markov to $G$ if $d$-separation implies conditional independence
($X$ $d$-separated from $Y$ by $S\Rightarrow X\indep Y\given S$), and faithful \citep{algo15} if the
converse holds as well. DAGs that encode the same $d$-separations form a Markov
equivalence class, summarized by a completed partially directed acyclic graph (CPDAG) \citep{def10}: the
shared skeleton (undirected adjacencies) together with the edge orientations common to the
whole class.

\paragraph{The PC algorithm.} PC~\citep{algo65} learns the CPDAG from
conditional independence queries in two phases. The skeleton phase begins with the
complete undirected graph and deletes an edge $X-Y$ as soon as some set $S$ drawn from the neighbours of
$X$ or $Y$ gives $X\indep Y\given S$; the orientation phase then orients $v$-structures and applies the Meek rules~\citep{algo55} to obtain the CPDAG. PC uses its
conditional independence test only as an oracle: if the test reports independence
exactly on the $d$-separations of $G$, and $P$ is Markov and faithful to $G$, PC
returns $\mathrm{CPDAG}(G)$. Moreover, any relation that coincides with $d$-separation on $(P,G)$
may replace full conditional independence without affecting the guarantee, a fact we use in
Section~\ref{sec:bankfaithful}. In finite samples the oracle is a statistical test, and
consistency of PC follows from consistency of the test across the queries it issues.

\paragraph{The Generalised Covariance Measure.} For continuous $X,Y$ and conditioning set
$Z$, write the regression residuals $e_X=X-\E[X\given Z]$ and $e_Y=Y-\E[Y\given Z]$.
GCM~\citep{ci32} tests $X\indep Y\given Z$ by studentizing the empirical mean of the product
$e_Xe_Y$; it has asymptotic level $\alpha$ under a product condition on the
estimation rates of nuisance functions (each $o_P(n^{-1/4})$), and is the nonparametric, symmetric
generalization of partial correlation, with population target the expected conditional
covariance $\E[\Cov(X,Y\given Z)]$. 
The weighted variant (WGCM,~\citealp{ci50}) stays within
the covariance family, while the projected variant (PCM,~\citealp{ci52}) targets the
conditional mean, testing $\E[Y\given X,Z]=\E[Y\given Z]$. GCM and WGCM thus have power zero
whenever $\Cov(X,Y\given Z)=0$, and PCM whenever the conditional mean is unchanged; all three
are blind to pure scale (location-scale) dependence, where $X$ alters the distribution of $Y$
but not its mean. GFCM removes this shared blind spot while keeping GCM's regression cost and validity machinery.

\section{Methods}
\label{sec:method}

\subsection{The Generalised Feature Covariance Measure}
\label{sec:gfcm}

GCM tests $X\indep Y\given Z$ through a single feature of each residual, the identity feature (the residual itself).
GFCM keeps that template but replaces the identity by a configurable set of
features chosen so the test sees beyond the conditional covariance. For a variable $A$
and conditioning set $Z$, let $e_A=A-\widehat\E[A\given Z]$ be the cross-fitted mean
residual. Let $c_A=|e_A|-\widehat\E[|e_A|\given Z]$ be the centered scale residual, using the
absolute residual by default for heavy-tail robustness (the squared form is an admissible
alternative); it has mean zero given $Z$ by construction, which the centering
Lemma~\ref{lem:orth} requires since the uncentered scale term is biased. Finally, let $r_\tau(A;Z)=\tau-\mathbf 1\{A\le
\widehat Q_\tau(A\given Z)\}$ be the quantile-indicator residual ($\E[r_\tau\given
Z]=0$ at the true conditional quantile; equivalently, under the location-scale model
$\widehat Q_\tau(A\given Z)=\widehat m_A(Z)+\widehat s_A(Z)\,\widehat q_\tau$ this is the
standardized form $\tau-\mathbf 1\{e_A/\widehat s_A\le\widehat q_\tau\}$ used in
Algorithm~\ref{alg:qgcm}). The set is
\[
\Phi(A;Z)=\big\{\,e_A,\;c_A,\;r_{\tau_1}(A;Z),\dots,r_{\tau_k}(A;Z)\,\big\},
\]
with default levels $\mathcal T=\{0.1,0.5,0.9\}$; for a categorical $A$ the set is the centered one-hot residuals $\mathbf 1\{A=c\}-\widehat\E[\mathbf 1\{A=c\}\given
Z]$.

From the admissible pairs (one feature from $\Phi(X;Z)$, one from $\Phi(Y;Z)$, at
least one with mean zero given $Z$) GFCM forms three blocks of per-observation
products: the moment block $\{e_Xe_Y,\,c_Xe_Y,\,e_Xc_Y\}$ and the two quantile
orientations $\{e_X\,r_\tau(Y)\}_\tau$ and $\{r_\tau(X)\,e_Y\}_\tau$. Every product in the
default set pairs a higher-order feature of one variable with the mean residual of the
other; products pairing two higher-order features (such as $c_Xc_Y$) are admissible under
the validity theory but excluded from the default set: they target a distinct and narrower
class of co-dispersion alternatives (conditional co-volatility and joint tail dependence) rather
than the location, scale, and shape dependence GFCM targets, and are better activated as a
domain-specific configuration (Section~\ref{sec:discussion}, item ix).

Each block is
reduced to one $p$-value by a multivariate GCM quadratic form, calibrated by its asymptotic
$\chi^2$ reference by default (a sign-flip bootstrap is exact in finite samples,
Lemma~\ref{lem:signflip}), and
the three block $p$-values are combined by the Cauchy (ACAT) rule~\cite{ci99}.
Pooling coherent evidence within a block avoids the dilution of treating every product
separately, while keeping the orientations in distinct blocks preserves validity when one
is structurally null (Proposition~\ref{prop:f1}).

The members target distinct dependence:
$e_Xe_Y$ is the conditional
covariance; $c_Xe_Y$ and $e_Xc_Y$ catch non-monotone mean and
location-scale dependence (where $\Cov(X,Y\given Z)=0$); the quantile products
$e_X\,r_\tau(Y)$ catch dependence in the conditional distribution at a chosen
level, including changes of tail shape that preserve the mean and variance, which the moment features miss. Ordinary GCM is the special case of a single feature. No CI test controls
size and has power against every alternative~\cite{ci32}, and GFCM does not evade this
barrier. Instead, the set fixes a detection class, and the detection class
characterisation (Section~\ref{sec:theory}) states which alternatives a given set
catches and provably misses, so one configures the set for the dependence of interest
(moments for mean/scale, quantiles for tails) rather than claiming universality; the choice
of features changes the test's power, not its level (Corollary~\ref{cor:validinv}).

\paragraph{Nuisance functions.} Validity requires the nuisances estimated at a
nonparametric rate (Section~\ref{sec:theory}); a fixed parametric design of low degree does
not achieve it and distorts size under misspecification
(Section~\ref{sec:nuisance}). GFCM's default nuisance is a cubic spline basis whose knot count grows with $n$,
$K(n)\approx 10\,n^{1/5}$ (rounded; floored at $25$ for small $n$, capped at $150$), on
rank-transformed $Z$: nonparametric, so its approximation error vanishes as $n$ grows and
orthogonality holds, yet every fit is ordinary ridge least squares on that basis.
The growing rule is what makes assumption (A1) bind in practice rather than asymptotically:
a fixed basis leaves an approximation bias the $\sqrt n$ statistic amplifies, drifting size
upward on a misspecified mean, while $K(n)$ holds the level at regression cost across $n$
(Section~\ref{sec:nuisance}).

The conditional mean $\widehat\E[A\given Z]$
and the conditional scale $\widehat s(Z)=\widehat\E[|e_A|\given Z]$ (the same fit that
defines $c_A$) are two such regressions; the conditional
quantile is then a location-scale model
$\widehat Q_\tau(A\given Z)=\widehat\E[A\given Z]+\widehat s(Z)\,\widehat q_\tau$, with
$\widehat q_\tau$ the empirical $\tau$-quantile of the standardized residual
$e_A/\widehat s(Z)$, so the quantile feature reuses the
mean and scale fits and needs no separate quantile regression (the squared-residual
variant instead standardizes by $\widehat s(Z)^2=\widehat\E[e_A^2\given Z]$). 

This makes the
same backend both fast (regression cost, scaling in $n$ to $10^5$:
Section~\ref{sec:speed}, Figure~\ref{fig:speed}) and robust against the bias that a fixed
parametric design incurs when misspecified. The conditioning design augments the rank-$Z$ spline with three blocks that the rank transform
alone cannot represent. First, a raw standardized linear-$Z$ block, which restores the scale the
rank transform discards and is what holds the level on heavy-tailed conditioners (the rank-$Z$
spline cannot represent a heavy-tailed conditioner's linear effect). Second, for $2\le|S|\le20$,
a polynomial of degree 3 interaction block, which protects validity against bilinear and product
confounders the additive spline misses, dropped past $|S|=20$ where the test is impractical
regardless. Third, at large $n$ with $|S|\ge2$, a cross-fitted single-index block, a
spline of degree 3 on each estimated index $\hat a^\top Z$ with $\hat a$ the average derivative
(ridge OLS) direction of $X$ and of $Y$ on $Z$, which captures non-additive means such as
$\sin(\sum_j Z_j)$ that neither the additive spline nor the interaction block of low order fits. The
single-index block is gated to large $n$,\footnote{Operationally the single-index block is
activated at $n\ge5000$ with $|S|\ge2$. The threshold comes from a crossover study on a
non-additive single-index mean ($\sin(\sum_j Z_j)$ at depth $|S|=3$): the single-index design
ties the additive design at $n=2000$ and overtakes it from $n\approx4000$, bounding the Type-I
rate to $0.04$--$0.13$ while the purely additive design's size climbs from $0.10$ to $0.47$ to
$0.99$, and reaching nominal ($\approx0.06$) by $n\approx2\times10^{4}$. The gate at $n=5000$ is a
safe margin above the crossover.} where the index direction is estimable, and reduces to
the additive path below the gate. The ridge penalty (selected by GCV on the scale regression,
a fixed small value on the mean regression, since GCV on the mean overfits the
heteroscedastic nulls) shrinks these blocks to negligible weight when they are not needed. The location-scale quantile residual is exactly mean zero given $Z$ under a
location-scale null, and its bias is otherwise second-order, like any nuisance error
(Section~\ref{sec:theory}). 

One caveat: the location-scale form is itself a shape
assumption. When the true conditional shape is not location-scale, level is still
controlled (the bias is second-order), but the estimand drifts from the population
$\tau$-quantile, so the clean effect curve reading of Proposition~\ref{thm:detection}
becomes approximate.

\paragraph{Calibration.} Each block's multivariate GCM quadratic form is calibrated by
default against its asymptotic $\chi^2_k$ reference (Theorem~\ref{thm:validity}), which
needs no resampling and is the calibration we run. A sign-flip (Rademacher) bootstrap of
its per-observation scores is the finite-sample alternative (Lemma~\ref{lem:signflip}):
exact under conditional sign symmetry (asymptotic, not exact, under conditional skew,
the regime it targets in the tail), robust to the unbounded moment features, and still
valid when a block covariance is degenerate. The two agree at every sample size we tested.

The Cauchy
(ACAT) combination~\cite{ci99} of the three block $p$-values is asymptotically
valid under arbitrary dependence between blocks and avoids both Bonferroni
conservativeness and a degenerate joint Wald: a pure scale edge makes one orientation
exactly null (Proposition~\ref{prop:f1}), which would inject a block with zero signal into a
single joint quadratic form, whereas a separate $p$-value per block leaves the other
blocks unaffected.

\subsection{Running GFCM inside PC}
\label{sec:fixes}

A distribution test that works on a single triple can still fail inside PC, where
an edge is deleted if any conditioning set yields independence. Two issues arise:
one specific to our asymmetric construction (F1), one shared by any fast residual test with a fixed nuisance (F2). GFCM absorbs both.

\paragraph{Directional asymmetry (F1).} The cross-products are not
symmetric in $(X,Y)$: $e_X\,r_\tau(Y)$ residualizes $X$ at the mean and $Y$ at a
quantile. For a pure scale edge $X\to Y$ (with $Y$ having mean zero given $X$), the orientation
that residualizes the driven variable $Y$ at the mean has population signal exactly zero
at every level and for any nuisance (Proposition~\ref{prop:f1}), $Y$ is uncorrelated
with every function of $X$, so used alone it falsely reports independence and PC deletes
the edge. This risk of deleting a true edge is specific to a mean-quantile cross-product; a
symmetric test (vanilla GCM, RCoT), having no preferred orientation, is not exposed to it.

We keep the asymmetric construction not to patch a
self-made defect but because it targets a directed, level-driven alternative, ``$X$'s
level drives $Y$'s distribution'', that a symmetric quantile-quantile product cannot: on the skew edge that preserves mean and variance the asymmetric product has power $1.00$ versus
$0.22$ for the symmetric one, which instead catches tail co-movement
(Section~\ref{sec:ablation}, Table~\ref{tab:asym}); the per-quantile effect curve and targeting at tail levels follow from the same factor (Section~\ref{sec:theory}). 

F1 is a consequence of that directed targeting; running both orientations
($e_X\,r_\tau(Y)$ and $r_\tau(X)\,e_Y$) removes it, detecting the
edge whichever way it is driven.

\paragraph{Nuisance fragility (F2).} A fast, fixed nuisance is fragile in two ways that
both surface inside PC. When a conditioning set contains a heavy-tailed variable, a spline
with knots spaced over the raw range wastes its resolution in the tails and underfits the
bulk mean, so the leaked residual inflates the null size (to $1.00$ on a Cauchy conditioner;
Table~\ref{tab:rankz}). And when the conditional mean is wiggly and non-polynomial, a
basis of fixed size cannot track it, and the approximation bias, amplified by $\sqrt n$, drives
level toward $1.00$ as $n$ grows (Section~\ref{sec:nuisance}).

Both are failures of a fast fixed
nuisance, not of the quantile target, and both are removed by GFCM's default growing-knot
spline on rank-$Z$ at regression cost. Unlike F1, this is not specific to our
test: any CI test based on regression residuals with a cheap fixed nuisance inherits it.
Together these take recovery of scale edges inside PC from $0\%$ to full recovery in our
experiments (Section~\ref{sec:experiments}). Algorithm~\ref{alg:qgcm} summarizes GFCM.

\begin{algorithm}[t]
\caption{GFCM conditional independence test}
\label{alg:qgcm}
\begin{algorithmic}[1]
\Require $X,Y$, each real-valued or categorical; mixed-type conditioning set $Z$; levels
$\mathcal{T}$, cross-fit folds $K_{\mathrm{fold}}$. A categorical variable uses its centered one-hot residual
set in place of $e,c,r_\tau$ (Section~\ref{sec:mixed}); the steps below show the
path for continuous variables.
\State $K_n \gets \min\{150,\,\max(25,\,\mathrm{round}(10\,n^{1/5}))\}$ \Comment{spline knots grow with $n$ (A1)}
\State $\tilde Z \gets \textsc{SplineBasis}_{K_n}(\textsc{RankTransform}(Z))$, with a raw standardized linear-$Z$ block and, for $2\le|Z|\le20$, a degree-3 interaction block \Comment{nonparametric design}
\State \textbf{if} $|Z|\ge2$ and $n\ge n_0$ (default $n_0{=}5000$) \textbf{then} append to $\tilde Z$ a cross-fitted single-index spline on $\widehat a^\top Z$, with $\widehat a$ the per-fold ridge-OLS (Stein) index direction \Comment{corrects non-additive single-index means at large $n$}
\State $e_X,e_Y \gets$ cross-fitted mean residuals (ridge with fixed penalty on $\tilde Z$)
\State $\alpha \gets$ ridge penalty selected by GCVfor the scale fit \Comment{data-adaptive}
\State $c_X\gets |e_X|-\widehat\E_\alpha[|e_X|\mid\tilde Z]$,\; $c_Y\gets |e_Y|-\widehat\E_\alpha[|e_Y|\mid\tilde Z]$ \Comment{centered scale ($|e|$ default, or $e^2$)}
\State $\widehat s_X,\widehat s_Y \gets$ cross-fitted conditional scale (from the scale fit) \Comment{for standardization}
\State $r_\tau(X),r_\tau(Y) \gets \tau-\mathbf 1\{e/\widehat s\le \widehat q_\tau\},\ \tau\in\mathcal T$ \Comment{quantile residuals}
\State $B_1\gets\{e_Xe_Y,\,c_Xe_Y,\,e_Xc_Y\}$ \Comment{moment block}
\State $B_2\gets\{e_X\,r_\tau(Y)\}_{\tau}$,\; $B_3\gets\{r_\tau(X)\,e_Y\}_{\tau}$ \Comment{quantile orientations}
\State for each block $B_j$:\; $p_j \gets$ analytic $\chi^2_k$ $p$-value of its multivariate GCM quadratic form (or sign-flip)
\State \Return $\textsc{Cauchy}(p_1,p_2,p_3)$ \Comment{ACAT}
\end{algorithmic}
\end{algorithm}

For a categorical variable the mean/quantile features are replaced by the
centered one-hot residuals (Section~\ref{sec:mixed}); the
quadratic transforms are dropped (redundant for a two-valued residual), and the
continuous side keeps its set.

\subsection{Mixed type data}
\label{sec:mixed}

GFCM is mixed-type through the same feature set. For a categorical variable the set is
the centered one-hot residuals $\mathbf 1\{A=c\}-\widehat\E[\mathbf 1\{A=c\}\given
Z]$, with $\widehat\E[\mathbf 1\{A=c\}\given Z]$ estimated by a cross-fitted ridge least-squares fit of the
indicator on the same spline design as the continuous nuisances. The GCM product needs only a consistent
estimate of the conditional mean $\E[\mathbf 1\{A=c\}\given Z]=\Pr(A=c\given Z)$, not a logistic
link; the least-squares fit supplies this while sharing the per-fold Cholesky factorization,
so a $J$-level categorical residualizes at the cost of $J-1$ extra least-squares targets and is no
more expensive than a continuous variable. (Fitted values may fall outside $[0,1]$; only the
residual enters the statistic, and the sign-flip / $\chi^2$ calibration is unaffected.) These
residuals have mean zero given $Z$, so the validity argument
(Section~\ref{sec:theory}) applies: it uses only that one factor of each product has
conditional mean zero, regardless of variable type, and the bounded categorical residual
meets the moment conditions automatically, with the density condition (A3) vacuous in the
absence of a continuous quantile feature. The quadratic transforms are
dropped for a binary residual (where they are redundant); a continuous variable keeps
its full set. 

For a $J$-level categorical, against each feature of the
other variable we test the $(J-1)$-vector of products with a single multivariate GCM
quadratic form (a sum-of-squares aggregation of the per-level statistics, $\chi^2$
calibrated), and combine these per-feature $p$-values by Cauchy.

This sum-of-squares pooling of the $J-1$ levels contrasts with the
original GCM's max-norm aggregation of multivariate residuals~\citep[\S3.2]{ci32}: a
categorical location shift moves many of the $J-1$ levels together, a dense signal the
quadratic form pools rather than reduces to a maximum.

The conditioning design is itself aware of variable type, continuous $Z$-columns enter
through the spline basis, categorical $Z$-columns through one-hot encoding (splining a categorical variable coded as integers is a calibration trap). Unlike covariance tests on dummy-coded
categoricals, the centered one-hot residual is invariant to relabeling the categories.

\section{Theoretical properties}
\label{sec:theory}

\subsection{Validity}
Each GFCM component is a product $\psi=\phi_X(X;Z)\,\phi_Y(Y;Z)$ in which at least one
factor is a residual that is conditionally centered (mean zero given $Z$) under the null.
The two factors are then conditionally independent given $Z$, so $\E[\psi\given Z]=0$ and
every component score has conditional mean zero (Lemma~\ref{lem:meanzero}). The centered factor also delivers
Neyman orthogonality: a perturbation of the other factor's nuisance (the mean
$\widehat\E[A\given Z]$, or, for a quantile factor, $\widehat Q_\tau$, which pulls in a
conditional density term) enters multiplied by that centered factor and so has zero
first-order effect. With cross-fitting, the surviving remainder is the product of
the two nuisance estimation errors, which is $o_P(n^{-1/2})$ when both are estimated at
$o_P(n^{-1/4})$ (Theorem~\ref{thm:validity}); each block's quadratic form statistic is then
asymptotically calibrated against its $\chi^2_k$ reference, and the Cauchy combination is asymptotically valid under
arbitrary dependence between blocks. The quantile and centered one-hot residuals are bounded
($|r_\tau|\le1$), so heavy-tailed or zero-inflated data does not threaten the component
CLTs.

The rate condition is the assumption that does the work here. The orthogonal score removes the
first-order nuisance error, leaving the product remainder above, which vanishes only if
both nuisances are consistent. A fixed parametric design (a low-degree polynomial) is
generally misspecified: it converges to a pseudo-true nuisance at which the factor no longer has mean zero, so the score acquires a bias and the test distorts size, not
merely power (Section~\ref{sec:nuisance} measures a null rejection rate up to $0.98$).
The spline basis is nonparametric: its approximation error vanishes as the knot count
grows with $n$, so orthogonality is restored, and this comes at regression cost.

Theorem~\ref{thm:validity} gives pointwise asymptotic
level, the standard guarantee in the constraint-based CI testing literature: KCI~\citep{ci24}
and RCoT~\citep{ci9} likewise establish pointwise asymptotic null distributions and handle
the regression nuisance at the level of consistency. Our hypotheses are, if anything,
stronger: (A2) imposes the double machine learning rate $o_P(n^{-1/4})$ rather than mere
consistency, and the sign-flip calibration (Lemma~\ref{lem:signflip}) is finite-sample exact
under conditional sign symmetry and otherwise converges to its $\chi^2$ reference at the
$O_P(n^{-1/2})$ Berry--Esseen rate. A uniformly valid version, uniform asymptotic
level over a class of nulls with explicit size rates, follows the GCM template
of~\citet{ci32}; we state the pointwise result, which suffices for the discovery guarantee
of Section~\ref{sec:bankfaithful}.

\subsection{Detection class of a feature set}
\label{sec:detection}
For a feature pair $(\phi_X,\phi_Y)$ define the effect $\beta_{\phi_X\phi_Y}=\E[\phi_X(X;Z)\,\phi_Y(Y;Z)]$, zero under the null. GFCM is consistent against an alternative iff
$\beta_{\phi_X\phi_Y}\neq0$ for some pair in the set, and powerless iff every
$\beta_{\phi_X\phi_Y}=0$. This fixes a precise detection class: the covariance pair detects
linear mean dependence; the quadratic pairs $c_Xe_Y,\,e_Xc_Y$ detect
non-monotone mean ($\Cov(X,X^2)=0$) and location-scale dependence; the quantile pairs
detect dependence in the conditional distribution at the chosen levels, including
changes of tail shape that preserve mean and variance and that every moment pair misses.
Consistent with the hardness of CI testing~\cite{ci32}, GFCM is not universal: it has
power zero against any alternative orthogonal to every feature in the set, dependence
that lives between the chosen quantile levels and in moments beyond the second. Testing a
finite set realizes Daudin's weak, or mean, conditional independence~\citep{ci107}:
the vanishing of $\E[\phi_X\phi_Y]$ over a fixed feature set is necessary but not sufficient
for full conditional independence, so this blind spot is intrinsic to any set, not an
artefact of ours. The
quantile part is characterized exactly by the quantile effect curve
$\beta(\tau) = \E\!\big[e_X\,(\tau-\mathbf{1}\{Y\le Q_\tau(Y\given Z)\})\big]$.

\begin{proposition}[Detection class]
\label{thm:detection}
Under the regularity conditions of the validity result, GFCM with feature set $\Phi$ is
consistent against an alternative iff $\beta_{\phi_X\phi_Y}\neq0$ for some admissible pair
$(\phi_X,\phi_Y)$, and has asymptotic power equal to its level iff $\beta_{\phi_X\phi_Y}=0$ for
all of them. For the quantile subset at levels $\mathcal{T}=\{\tau_1,\dots,\tau_k\}$
this specializes to the effect curve $\beta(\tau)$: (i) a location shift is
detected at any single level, including the median; (ii) for a scale shift with
symmetric noise, $\beta(0.5)=0$ and $\beta$ is odd around $0.5$, so $\{0.5\}$ is
blind but any $\tau\neq0.5$ detects it (under skewed noise the median also moves, so
$\{0.5\}$ may detect it as well); (iii) a finite set misses exactly the alternatives whose
effect curve vanishes at all chosen levels but not between them. The moment pairs add
their own coordinates: $e_Xe_Y$ detects linear mean dependence, and $c_Xe_Y,\,e_Xc_Y$
detect non-monotone mean and location-scale dependence where $\Cov(X,Y\given Z)=0$.
\end{proposition}

Proposition~\ref{thm:detection} justifies the default $\{0.1,0.5,0.9\}$ (the median
detects a location shift directly, while the symmetric tail pair $\{0.1,0.9\}$ catches
scale and shape dependence, to which the median is blind), states the blind spot
honestly, and gives a design rule (match the grid to the shape of interest, adding an
intermediate level when a specific alternative calls for it). It is a
statement about the power geometry of finite quantile sets, distinct from the
DML validity machinery~\cite{ci25}. Proof sketches are in
Appendix~\ref{app:proofs}; the claims are validated empirically
(Figure~\ref{fig:detection}).

We are careful about what this dial buys. In unsupervised discovery one cannot tune $\tau$
per edge: GFCM ships a fixed default grid and inherits its blind spots, so the configurability
is a lever for supervised, targeted analysis (a known tail level of interest), not an advantage at discovery time over an undirected operator such as SGCM. ``No expert knowledge''
here means only that GFCM needs no model of where the heteroscedasticity lives, unlike
ParCorr-WLS.

\begin{figure}[ht]\centering
\includegraphics[width=0.82\textwidth]{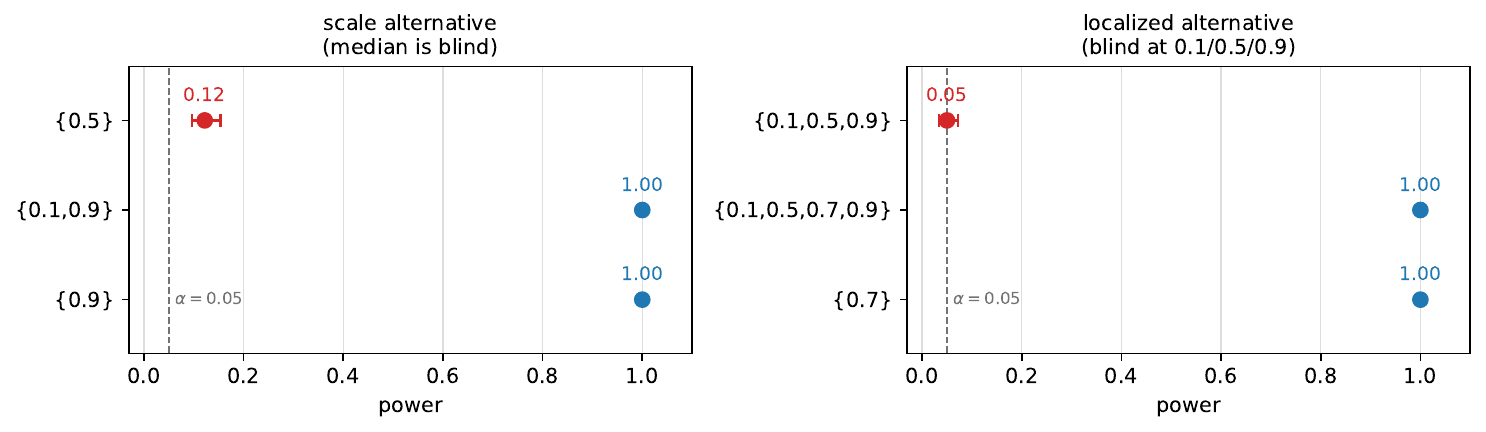}
\caption{\textbf{Detection class of the quantile subset.} $n=2000$, $500$
replications; points are rejection rates, whiskers are Wilson $95\%$ intervals.
Both panels isolate the single decisive level. Left: a scale alternative is invisible
to the median $\{0.5\}$ (power $0.12$) but caught by any grid containing a tail level,
including $\{0.9\}$ alone ($1.00$). Right: an alternative whose effect curve vanishes at
$0.1,0.5,0.9$ escapes that grid (power $0.05$) yet is recovered by adding $\tau=0.7$
($1.00$) or testing $\{0.7\}$ alone ($1.00$). Red markers flag the blind grids; the grey
line is the nominal level $\alpha=0.05$.}
\label{fig:detection}
\end{figure}

\subsection{$\Phi$-faithfulness}
\label{sec:bankfaithful}
Inside PC, GFCM does not test full conditional independence; it tests the configured
null that every set cross-covariance vanishes. The object PC then recovers is an
identification target defined by the set, and naming it precisely is what lets a
consistency statement go through. Write $X\indep_\Phi Y\given Z$ for $\Phi$-independence: $\beta_{\phi_X\phi_Y}(X,Y\given Z)=0$ for every admissible pair $(\phi_X,\phi_Y)$ in
the set (Section~\ref{sec:gfcm}). True conditional independence implies $\Phi$-independence,
$X\indep Y\given Z\Rightarrow X\indep_\Phi Y\given Z$, but not conversely: a dependence
living entirely in the set's blind spot (Proposition~\ref{thm:detection}(iii)) satisfies
$\indep_\Phi$ while violating $\indep$.

\begin{definition}[$\Phi$-faithfulness]
\label{def:phifaith}
A distribution $P$ is $\Phi$-faithful to a DAG $G$ if, for all disjoint $X,Y,Z$,
$X\indep_\Phi Y\given Z$ holds if and only if $X$ is $d$-separated from $Y$ by $Z$ in $G$.
\end{definition}

Definition~\ref{def:phifaith} is the $\Phi$-relative form of the faithfulness assumption every
constraint-based method makes relative to the independence relation its test actually checks.
Linear-Gaussian PC with partial correlation assumes exactly this for the covariance relation,
that a vanishing partial correlation coincides with $d$-separation (linear
faithfulness; \citealp{algo15,algo34}); $\Phi$-faithfulness generalizes it from the
single covariance relation to the detection class of an arbitrary features set. The forward
direction is free: under the global Markov property $d$-separation implies
$\indep$, hence $\indep_\Phi$. The reverse direction is the assumption, and it is
stronger than ordinary faithfulness precisely because $\indep_\Phi$ is the coarser
relation: it requires every $d$-connection to leave a signal somewhere in the set,
not merely some conditional dependence to exist. An edge whose dependence is confined to the
set's blind spot violates $\Phi$-faithfulness.

\begin{theorem}[$\Phi$-CPDAG consistency]
\label{thm:phicpdag}
Let $P$ be global Markov and $\Phi$-faithful to a DAG $G$. Then PC using $\indep_\Phi$ as its
conditional independence oracle returns the CPDAG of $G$. With GFCM as the test at a level
sequence $\alpha_n\to0$, valid under the conditions of Theorem~\ref{thm:validity} and
consistent against every $\Phi$-detectable separation (Proposition~\ref{thm:detection}), the
empirical output converges in probability to $\mathrm{CPDAG}(G)$.
\end{theorem}

\begin{proof}[Proof sketch]
As recalled in Section~\ref{sec:prelim}, PC's soundness and completeness depend on the oracle
only through the requirement that it report independence exactly on the $d$-separations of
$G$; given such an oracle the algorithm is a finite combinatorial procedure returning
$\mathrm{CPDAG}(G)$~\citep{algo15,algo65}.
Under global Markov, $d$-separation implies $\indep_\Phi$; under $\Phi$-faithfulness,
$\indep_\Phi$ implies $d$-separation; so $\indep_\Phi$ coincides with $d$-separation on
$(P,G)$ and the population claim follows. For the finite-sample claim, PC issues a fixed
finite number of GFCM calls; each is a consistent test of its set null (level $\to\alpha_n$
by Theorem~\ref{thm:validity}, power $\to1$ on $\Phi$-detectable dependence by
Proposition~\ref{thm:detection}), so a union bound gives the empirical CPDAG equal to the
population one with probability $\to1$.
\end{proof}

\begin{corollary}[Conservative recovery without $\Phi$-faithfulness]
\label{cor:phiconservative}
If $P$ is global Markov but not $\Phi$-faithful to $G$, PC$+$GFCM still removes every true
non-adjacency (no spurious edges in the population limit) but may delete a true edge whose
dependence is $\Phi$-blind on some separating set; it then recovers the $\Phi$-skeleton,
a subgraph of the skeleton of $G$. The output is thus consistent for the $\Phi$-CPDAG,
which equals the causal CPDAG exactly when $P$ is $\Phi$-faithful.
\end{corollary}

The configurable set thus defines an identification target:
PC$+$GFCM is consistent for the CPDAG of the graph whose independence model is the
set's detection class, it never invents adjacencies the set cannot see, and a richer set
(more quantile levels, higher moments) enlarges the class of edges it can recover. Sharper identification of quantile CPDAG, characterizing which edges appear and vanish as the
level grid $\mathcal{T}$ varies, is left to the companion algorithmic work.

\section{Simulation studies}
\label{sec:experiments}

We evaluate GFCM against the panel below on synthetic data from a controlled generator,
reporting two sample-size regimes throughout: a small-sample regime ($n=2000$, where the
partial copula predecessor \citep{ci90} is included) and a large-sample regime ($n=10^5$, where
that test is omitted owing to its per-query cost). All power is reported size-corrected: each
test's rejection threshold is set to exact $5\%$ on its matched null, so no test can buy power by
over-rejecting. All cells are $500$ repetitions at $\alpha=0.05$, with Wilson
$95\%$ confidence intervals shown in brackets.

\subsection{Setup and the test panel}
\label{sec:setup}
The main panel comprises six tests: GFCM; the partial copula test \citep{ci90} (PartCopula), the quantile copula predecessor; BLITZ~\citep{ci100}, the
fast residualisation rival; FFCI~\citep{ci105}, the fast mixed-type Fourier feature test; RCoT~\citep{ci9}; and
gradient-boosted GCM~\citep{ci32}. Data are drawn from a structural
generator with two decoupled seeds, one fixing the graph and mechanisms and one the sample,
for exact reproducibility. We report size at $\alpha=0.05$, power, and, for the discovery
experiments, skeleton SHD and recall per edge type. The tests are grouped by what each can
see: the conditional mean, the conditional covariance, or the full conditional distribution.

\subsection{Calibration}
\label{sec:nuisance}
On a battery of conditional independence nulls with heavy-tailed, heteroscedastic,
mixed-type, and wiggly-mean conditioners, across conditioning depth $|S|\in\{1,3,5\}$,
GFCM and BLITZ are the only broadly size-valid tests (Table~\ref{tab:calib-n2k});
the partial copula test inflates catastrophically on heavy-tailed and heteroscedastic conditioners and is
degenerate on mixed-type ones, boosted GCM is chronically liberal, and RCoT breaks down at
depth. FFCI, run at its published defaults,\footnote{We run Tetrad's shipped \texttt{FfCi} at its
defaults, which match \citet{ci105}: one-hot discrete features ($\rho=0$), $m_{XY}=10$, $m_Z=100$,
ridge $\lambda=1$, median bandwidth, and the gamma $p$-value; we tune nothing. FFCI's original
evaluation is on continuous, Gaussian-noise data and reports discovery accuracy rather than test
size, and its authors describe its $p$-values as approximate screening tools rather than exact
tests \citep[\S10]{ci105}, so these mixed-type and heavy-tailed calibration cells probe a regime
outside that validation.} shows two distinct failures on these nulls: its median bandwidth Fourier
map over-rejects under heavy-tailed conditioners, and its default one-hot discrete featurization
over-rejects outright under a categorical conditioner (size $1.00$ on \texttt{mixed\_Z}).
GFCM is mildly liberal only at small $n$ and deep $|S|$ on the heteroscedastic and non-additive
nulls, a corner outside the operating region that self-corrects with sample size: in the large-sample regime
(Table~\ref{tab:calib-n1e5}) every null returns to nominal level, even as the fast
residualisation competitors instead deteriorate there (RCoT and boosted GCM reach size
$0.8$--$1.0$ at $|S|\ge3$). The \texttt{nonlin\_mean} null has a sinusoidal conditional mean
$\sin(1.5\,Z_1)$ that is additive in a single parent; GFCM's growing-knot spline fits it directly
and holds nominal level at every depth ($0.05$--$0.07$), while FFCI's fixed Fourier featurization
cannot track it (size $0.43$ and $0.89$ at $|S|=1,3$)

\begin{table}[t]\centering\small
\setlength{\tabcolsep}{4pt}
\resizebox{\textwidth}{!}{%
\begin{tabular}{llcccccc}
\toprule
mechanism & $|S|$ & GFCM & PartCopula & BLITZ & FFCI & RCoT & GCM-boost \\
\midrule
heavy\_tail & 1 & 0.05~[0.03--0.07] & \textbf{0.17~[0.14--0.20]} & 0.06~[0.04--0.08] & 0.13~[0.10--0.16] & 0.04~[0.03--0.06] & \textbf{0.27~[0.23--0.31]} \\
hetero & 1 & 0.06~[0.04--0.08] & 0.05~[0.03--0.07] & 0.05~[0.03--0.07] & 0.03~[0.02--0.05] & 0.05~[0.03--0.07] & 0.06~[0.05--0.09] \\
mixed\_Z & 1 & 0.04~[0.03--0.06] & -- & 0.04~[0.02--0.06] & \textbf{1.00~[0.99--1.00]} & 0.04~[0.03--0.07] & 0.04~[0.03--0.07] \\
nonlin\_mean & 1 & 0.06~[0.04--0.08] & 0.07~[0.05--0.10] & 0.04~[0.03--0.07] & 0.05~[0.04--0.08] & 0.06~[0.04--0.08] & 0.05~[0.03--0.07] \\
heavy\_tail & 3 & 0.07~[0.05--0.09] & \textbf{1.00~[0.99--1.00]} & 0.06~[0.05--0.09] & \textbf{1.00~[0.99--1.00]} & 0.11~[0.08--0.14] & 0.12~[0.09--0.15] \\
hetero & 3 & 0.06~[0.05--0.09] & \textbf{0.99~[0.98--1.00]} & 0.05~[0.03--0.07] & \textbf{1.00~[0.99--1.00]} & 0.05~[0.03--0.07] & 0.10~[0.08--0.13] \\
mixed\_Z & 3 & 0.04~[0.03--0.06] & -- & 0.06~[0.04--0.09] & \textbf{1.00~[0.99--1.00]} & 0.06~[0.04--0.09] & 0.11~[0.09--0.14] \\
nonlin\_mean & 3 & 0.10~[0.08--0.13] & 0.06~[0.04--0.08] & 0.11~[0.09--0.14] & 0.06~[0.05--0.09] & 0.07~[0.05--0.09] & 0.09~[0.06--0.11] \\
heavy\_tail & 5 & 0.06~[0.05--0.09] & \textbf{1.00~[0.99--1.00]} & 0.05~[0.03--0.07] & \textbf{1.00~[0.99--1.00]} & \textbf{0.91~[0.88--0.93]} & \textbf{0.25~[0.21--0.29]} \\
hetero & 5 & 0.09~[0.07--0.12] & \textbf{1.00~[0.99--1.00]} & 0.06~[0.04--0.08] & \textbf{1.00~[0.99--1.00]} & \textbf{0.37~[0.32--0.41]} & 0.11~[0.08--0.14] \\
mixed\_Z & 5 & 0.06~[0.05--0.09] & -- & 0.05~[0.03--0.07] & \textbf{1.00~[0.99--1.00]} & 0.14~[0.11--0.18] & 0.11~[0.08--0.14] \\
nonlin\_mean & 5 & 0.07~[0.05--0.10] & 0.05~[0.03--0.07] & 0.06~[0.04--0.09] & 0.05~[0.04--0.08] & 0.05~[0.03--0.07] & 0.07~[0.05--0.10] \\
\bottomrule
\end{tabular}}
\caption{Calibration on hard nulls, $n=2000$ (size; Wilson 95\% CI). GFCM and BLITZ are the only
broadly valid tests. FFCI is at its published defaults (see text): it is anticonservative under
heavy-tailed conditioners, a tunable bandwidth default that a wider bandwidth would recalibrate at
the cost of the tail power it already lacks (Figure~\ref{fig:tail-depth}), and its size rises to
$1.00$ under a categorical conditioner through its default one-hot discrete featurization.}
\label{tab:calib-n2k}
\end{table}

\begin{table}[t]\centering\small
\setlength{\tabcolsep}{4pt}
\resizebox{\textwidth}{!}{%
\begin{tabular}{llccccc}
\toprule
mechanism & $|S|$ & GFCM & BLITZ & FFCI & RCoT & GCM-boost \\
\midrule
heavy\_tail & 1 & 0.06~[0.04--0.09] & 0.08~[0.06--0.11] & \textbf{1.00~[0.99--1.00]} & 0.06~[0.04--0.09] & \textbf{0.93~[0.91--0.95]} \\
hetero & 1 & 0.05~[0.03--0.07] & 0.04~[0.03--0.07] & 0.09~[0.07--0.12] & 0.05~[0.04--0.08] & 0.04~[0.03--0.06] \\
mixed\_Z & 1 & 0.05~[0.04--0.08] & 0.05~[0.03--0.07] & \textbf{1.00~[0.99--1.00]} & 0.04~[0.03--0.07] & 0.05~[0.03--0.07] \\
nonlin\_mean & 1 & 0.05~[0.03--0.07] & 0.04~[0.03--0.06] & \textbf{0.43~[0.38--0.47]} & 0.04~[0.03--0.07] & 0.05~[0.03--0.07] \\
heavy\_tail & 3 & 0.07~[0.05--0.09] & 0.06~[0.04--0.08] & \textbf{1.00~[0.99--1.00]} & \textbf{0.97~[0.95--0.98]} & \textbf{1.00~[0.99--1.00]} \\
hetero & 3 & 0.05~[0.03--0.07] & 0.07~[0.05--0.10] & \textbf{1.00~[0.99--1.00]} & 0.12~[0.09--0.15] & 0.07~[0.05--0.09] \\
mixed\_Z & 3 & 0.07~[0.05--0.10] & 0.04~[0.03--0.06] & \textbf{1.00~[0.99--1.00]} & 0.05~[0.03--0.07] & 0.05~[0.04--0.08] \\
nonlin\_mean & 3 & 0.07~[0.05--0.10] & 0.04~[0.03--0.06] & \textbf{0.89~[0.86--0.92]} & 0.07~[0.05--0.09] & 0.04~[0.03--0.07] \\
heavy\_tail & 5 & 0.06~[0.04--0.08] & 0.05~[0.04--0.08] & \textbf{1.00~[0.99--1.00]} & \textbf{1.00~[0.99--1.00]} & \textbf{1.00~[0.99--1.00]} \\
hetero & 5 & 0.05~[0.04--0.08] & 0.05~[0.04--0.08] & \textbf{1.00~[0.99--1.00]} & \textbf{0.87~[0.84--0.90]} & \textbf{0.21~[0.18--0.25]} \\
mixed\_Z & 5 & 0.06~[0.04--0.08] & 0.04~[0.03--0.07] & \textbf{1.00~[0.99--1.00]} & \textbf{0.74~[0.70--0.78]} & 0.12~[0.09--0.15] \\
nonlin\_mean & 5 & 0.06~[0.04--0.08] & 0.03~[0.02--0.05] & 0.07~[0.05--0.10] & 0.07~[0.05--0.09] & 0.08~[0.06--0.11] \\
\bottomrule
\end{tabular}}
\caption{Calibration, large-sample regime $n=10^5$ (size; Wilson 95\% CI). GFCM holds nominal
level across every null at every depth, while the fast residualisation competitors
deteriorate at scale (FFCI, RCoT and boosted GCM reach $0.8$--$1.0$ at $|S|\ge3$).
The \texttt{nonlin\_mean} null has a sinusoidal mean $\sin(1.5\,Z_1)$ additive in a single parent,
which GFCM's growing-knot spline fits directly (size $0.05$--$0.07$ at every depth) while FFCI's
fixed Fourier featurization does not (up to $0.89$).}
\label{tab:calib-n1e5}
\end{table}

\subsection{Power}
\label{sec:ablation}
Against alternatives that each perturb a single feature of the response (linear mean,
non-monotone mean, scale, and tail shape), GFCM matches the panel at ceiling on the standard
alternatives and is the only test that keeps power on the skew alternative as the
conditioning set deepens (Table~\ref{tab:power-n2k}). FFCI is competitive at $|S|{=}1$ ($0.92$)
but collapses toward chance by $|S|\ge3$, the other fast residualisation tests need far more data
to resolve it (reaching full power only by $n=10^5$), the partial copula test stays at chance (structurally blind to
skew), and boosted GCM is blind to everything beyond the mean. All power in this section is size-corrected: each test is re-thresholded to exact $5\%$
size on its matched null, so the comparison is free of calibration differences and no test can buy
power by over-rejecting.

\begin{table}[t]\centering\small
\setlength{\tabcolsep}{4pt}
\resizebox{\textwidth}{!}{%
\begin{tabular}{llcccccc}
\toprule
mechanism & $|S|$ & GFCM & PartCopula & BLITZ & FFCI & RCoT & GCM-boost \\
\midrule
linear\_mean & 1 & 1.00~[0.99--1.00] & 0.99~[0.97--0.99] & 1.00~[0.99--1.00] & 1.00~[0.99--1.00] & 1.00~[0.99--1.00] & 1.00~[0.99--1.00] \\
nonmonotone\_z2 & 1 & 1.00~[0.99--1.00] & 1.00~[0.99--1.00] & 1.00~[0.99--1.00] & 1.00~[0.99--1.00] & 0.99~[0.98--1.00] & 0.06~[0.04--0.08] \\
scale & 1 & 1.00~[0.99--1.00] & 0.99~[0.98--1.00] & 1.00~[0.99--1.00] & 1.00~[0.99--1.00] & 0.99~[0.98--1.00] & 0.06~[0.05--0.09] \\
tail\_shape & 1 & \textbf{0.99~[0.97--0.99]} & 0.06~[0.04--0.08] & 0.16~[0.13--0.19] & 0.92~[0.89--0.94] & 0.41~[0.37--0.46] & 0.08~[0.06--0.11] \\
linear\_mean & 3 & 1.00~[0.99--1.00] & 0.99~[0.98--1.00] & 1.00~[0.99--1.00] & 1.00~[0.99--1.00] & 1.00~[0.99--1.00] & 1.00~[0.99--1.00] \\
nonmonotone\_z2 & 3 & 1.00~[0.99--1.00] & 1.00~[0.99--1.00] & 1.00~[0.99--1.00] & 1.00~[0.99--1.00] & 1.00~[0.99--1.00] & 0.06~[0.04--0.08] \\
scale & 3 & 1.00~[0.99--1.00] & 0.99~[0.97--0.99] & 1.00~[0.99--1.00] & 1.00~[0.99--1.00] & 0.97~[0.95--0.98] & 0.05~[0.03--0.07] \\
tail\_shape & 3 & \textbf{0.90~[0.87--0.92]} & 0.05~[0.03--0.07] & 0.12~[0.10--0.15] & 0.21~[0.18--0.25] & 0.11~[0.08--0.14] & 0.06~[0.04--0.08] \\
linear\_mean & 5 & 1.00~[0.99--1.00] & 0.99~[0.98--1.00] & 1.00~[0.99--1.00] & 1.00~[0.99--1.00] & 1.00~[0.99--1.00] & 1.00~[0.99--1.00] \\
nonmonotone\_z2 & 5 & 1.00~[0.99--1.00] & 1.00~[0.99--1.00] & 1.00~[0.99--1.00] & 1.00~[0.99--1.00] & 0.99~[0.98--1.00] & 0.03~[0.02--0.05] \\
scale & 5 & 1.00~[0.99--1.00] & 0.97~[0.95--0.98] & 1.00~[0.99--1.00] & 1.00~[0.99--1.00] & 0.94~[0.91--0.95] & 0.04~[0.02--0.06] \\
tail\_shape & 5 & \textbf{0.86~[0.82--0.88]} & 0.07~[0.05--0.10] & 0.21~[0.17--0.24] & 0.15~[0.12--0.18] & 0.06~[0.05--0.09] & 0.03~[0.02--0.05] \\
\bottomrule
\end{tabular}}
\caption{Power across the alternative ladder, $n=2000$ (size-corrected power: each test's
threshold set to exact $5\%$ size on the matched null; Wilson 95\% CI). Bold marks GFCM's standout
advantage on the skew alternative, where it alone keeps power as $|S|$ grows (FFCI matches it
at $|S|{=}1$ but collapses by $|S|\ge3$).}
\label{tab:power-n2k}
\end{table}

\subsection{The conditional tail edge}
\label{sec:tailedge}
This is the centerpiece. We let $X$ drive the conditional skew of $Y$ with mean and variance
held fixed, and sweep the effect strength from weak to strong. After size-correcting every test to
exact $5\%$ on the matched null, GFCM's power rises from $0.29$ to $1.00$; FFCI is the nearest
competitor but trails ($0.13$ to $0.84$), while BLITZ (about $0.2$, roughly one rejection in five),
RCoT, and the partial copula test (at chance, structurally blind to skew) stay far lower (Figure~\ref{fig:tail-strength}).

\begin{figure}[t]\centering
\includegraphics[width=0.72\textwidth]{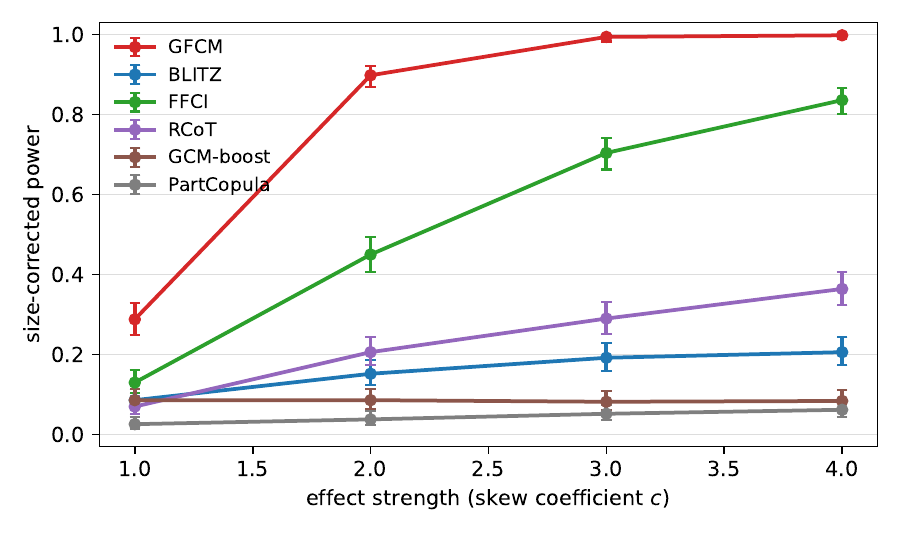}
\caption{\textbf{The conditional-tail edge.} Size-corrected power versus effect strength ($|S|{=}1$,
$n=2000$; each test's rejection threshold set to exact $5\%$ on its matched null, so the comparison
is calibration-free; error bars are Wilson 95\% intervals over $500$ repetitions). GFCM rises steeply
to full power; FFCI is the nearest competitor but trails, while BLITZ and RCoT rise only modestly and
the partial-copula test and boosted GCM stay near chance. Raw sizes are in Table~\ref{tab:calib-n2k}.}
\label{fig:tail-strength}
\end{figure}

The advantage is a finite-sample one. At $n=10^5$ the edge washes out: with enough data the general
random feature tests resolve the same conditional skew, so BLITZ, FFCI, and RCoT all reach
$0.9$--$1.0$ alongside GFCM (size correction also removes FFCI's mild large-$n$ liberality), and the
column is no longer separable. GFCM's distinctive value is therefore at the small to moderate sample
sizes typical of the high-dimensional regime, not asymptotically.

\subsection{Convergence with sample size}
\label{sec:convergence}
The two advantages above are not small-sample artifacts. The calibration advantage
grows with $n$ as the competitors' size diverges, and the power advantage grows with
conditioning depth $|S|$ (Figure~\ref{fig:conv}, with the full grid of nulls and
alternatives in Appendix~\ref{app:experiments}, Tables~\ref{tab:conv-null}
and~\ref{tab:conv-power}). On the heavy-tailed and mixed-type nulls (top), GFCM and
BLITZ hold near nominal at every depth as $n$ grows, while FFCI, the partial copula
test, and boosted GCM climb past $0.2$ and keep rising, their size diverging rather
than converging. On the skew edge (bottom), GFCM's size-corrected power dominates
throughout and its margin over the random-feature tests widens with the conditioning
depth $|S|$.

\begin{figure}[tbp]\centering
\makebox[\textwidth][l]{\textbf{A\quad Calibration vs $n$}}\\[0.2em]
\includegraphics[width=\textwidth]{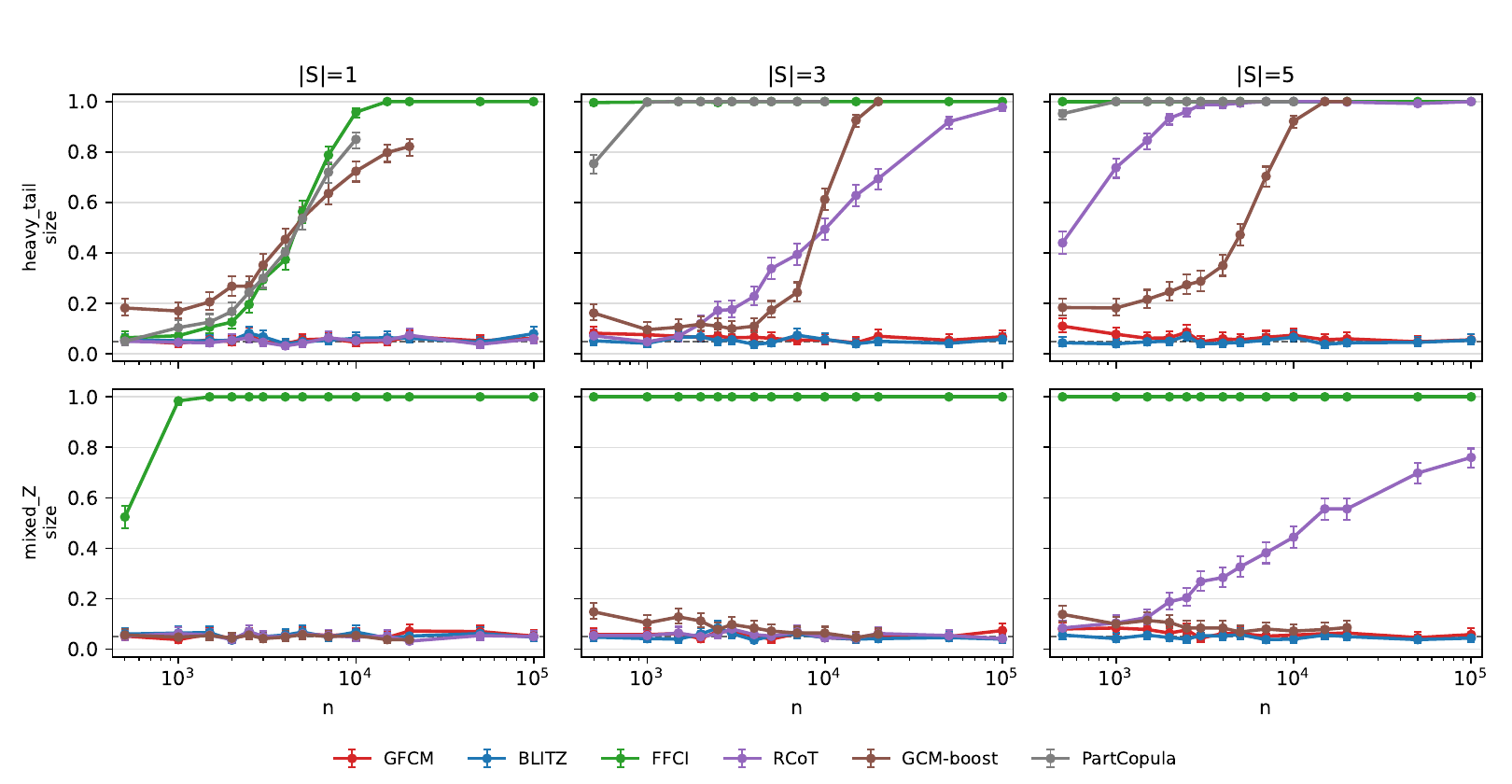}\\[0.7em]
\makebox[\textwidth][l]{\textbf{B\quad Size-corrected power}}\\[0.2em]
\includegraphics[width=\textwidth]{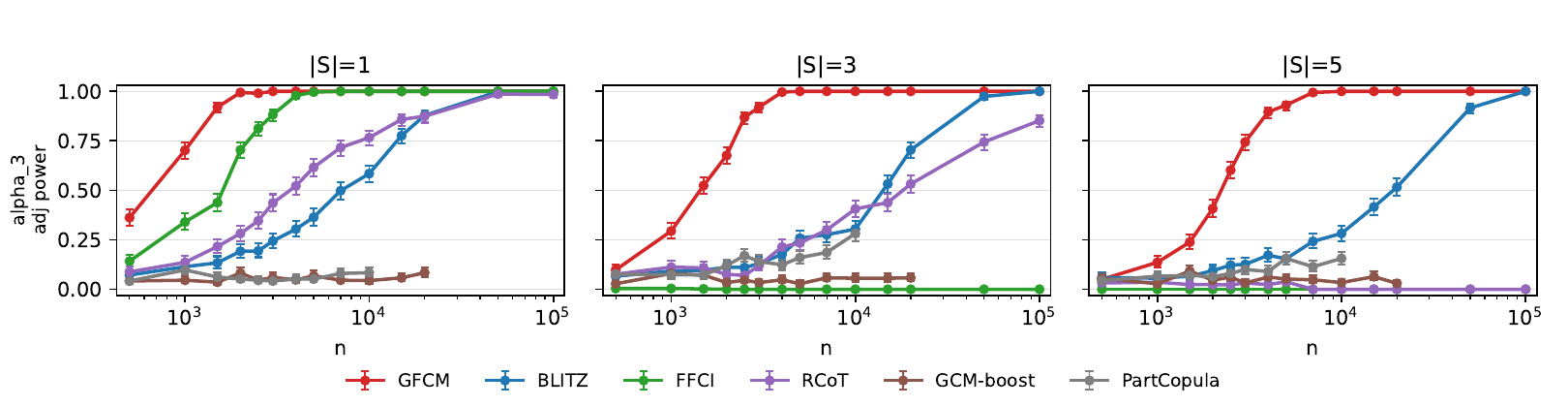}
\caption{Convergence with sample size. Columns are conditioning depth $|S|\in\{1,3,5\}$; curves are
$500$-repetition size or size-corrected power with Wilson bands. \textbf{Panel A (calibration):}
null rejection rate (Type I) versus $n$ on the heavy-tailed and mixed-type nulls. GFCM and BLITZ
stay near the nominal level, while FFCI, the partial copula test, boosted GCM, and (at depth) RCoT diverge.
\textbf{Panel B (power):} size-corrected power versus $n$ on the skew edge (\texttt{alpha\_3}).
GFCM dominates, and the gap widens with conditioning depth $|S|$. The partial copula test and boosted
GCM appear only up to the sample sizes at which they remain computationally feasible, not across the
full range of $n$: the partial copula test relies on an R routine based on resampling whose cost grows
steeply with $n$, and boosted GCM fits its nuisance functions by gradient boosting, which becomes
prohibitively slow at large samples, so we do not evaluate them at the largest $n$. The complete grid
of nulls and alternatives is tabulated in Tables~\ref{tab:conv-null} and~\ref{tab:conv-power}.}
\label{fig:conv}
\end{figure}

\subsection{A matched-calibration comparison to FFCI}
\label{sec:ffci}
The most recent fast mixed-type test is FFCI~\citep{ci105}, a residualisation test using random Fourier features in the RCoT family extended to discrete variables. As tests differ in
finite-sample calibration, a fair power comparison must equalise size first: we therefore
report size-corrected power, setting each test's rejection threshold to give exactly
$5\%$ on the matched null and reading power at that common level (Figure~\ref{fig:tail-depth}).
FFCI is the strongest of the random feature competitors on tails and is competitive
at shallow conditioning (power $0.92$  on the skew edge at $|S|{=}1$), but its power
collapses with conditioning depth ($0.92\!\to\!0.21\!\to\!0.15$ across
$|S|\in\{1,3,5\}$), whereas GFCM's targeted features degrade gracefully
($0.99\!\to\!0.90\!\to\!0.86$). The general Fourier feature map must be smoothed to stay
calibrated under heavy-tailed conditioners, which washes out precisely the tail signal the
targeted set retains; the advantage is thus structural, not a calibration artifact.

\begin{figure}[t]\centering
\includegraphics[width=0.62\textwidth]{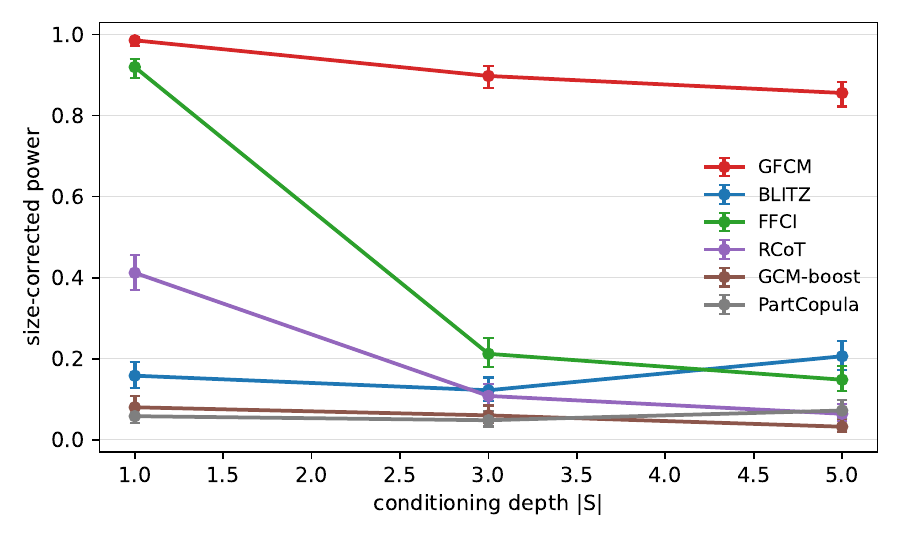}
\caption{\textbf{Graceful degradation under depth.} Size-corrected power on the skew edge versus
conditioning depth $|S|$ ($n=2000$; threshold set to $5\%$ on the matched null). At matched
calibration GFCM retains its tail power as the conditioning set grows ($0.99\to0.90\to0.86$),
while FFCI, competitive at $|S|{=}1$ ($0.92$), collapses toward chance ($0.21$, $0.15$); BLITZ and
RCoT stay low throughout. Error bars are Wilson 95\% intervals over $500$ repetitions. The
effect-strength view of this edge is in Figure~\ref{fig:tail-strength}.}
\label{fig:tail-depth}
\end{figure}

\subsection{Mixed-type data}
\label{sec:mixed-exp}
GFCM runs natively on categorical and continuous variables through the same feature set;
the partial copula test is continuous-only, and BLITZ is not native either: it handles categoricals only through a jittered integer coding, not invariant to relabeling, though it suffices here. We sweep the four $(X,Y)$
type pairings (C continuous, K categorical) on a spread alternative where $X$ drives the
conditional dispersion of $Y$, and report size-corrected power (Table~\ref{tab:mixed}). GFCM
keeps full power on every pairing, as does BLITZ. The other mixed-type tests fail in two distinct
ways. FFCI loses calibration outright whenever a continuous variable is present: its null size is
$1.0$, so size correction leaves it with zero discrimination on CC at $n=2000$ and on CC, CK, and
KC by $n=10^5$, retaining power only on the fully categorical KK pairing. The
mixed tests based on covariance (ci-mm~\citep{ci_mm}, boosted GCM) stay calibrated but are blind to the spread
edges beyond covariance: they catch the pairings with categorical response CK and KK, where the dependence shows in
class probabilities, but miss CC and KC, where $X$ drives the spread of a continuous $Y$ at zero
conditional mean dependence, the structure a covariance cannot see.

\begin{table}[t]\centering\small
\setlength{\tabcolsep}{5pt}
\resizebox{\textwidth}{!}{%
\begin{tabular}{lccccc}
\toprule
type pair & GFCM & FFCI & BLITZ & ci-mm & GCM-boost \\
\midrule
\multicolumn{6}{l}{\textit{$n=2000$}} \\
CC & \textbf{1.00~[0.99--1.00]} & 0.00~[0.00--0.01] & 1.00~[0.99--1.00] & 0.11~[0.09--0.14] & 0.05~[0.03--0.07] \\
CK & 1.00~[0.99--1.00] & 1.00~[0.99--1.00] & 1.00~[0.99--1.00] & 1.00~[0.99--1.00] & 1.00~[0.99--1.00] \\
KC & \textbf{1.00~[0.99--1.00]} & 1.00~[0.99--1.00] & 1.00~[0.99--1.00] & 0.06~[0.05--0.09] & 0.05~[0.03--0.07] \\
KK & 1.00~[0.99--1.00] & 1.00~[0.99--1.00] & 1.00~[0.99--1.00] & 1.00~[0.99--1.00] & 1.00~[0.99--1.00] \\
\midrule
\multicolumn{6}{l}{\textit{$n=10^5$}} \\
CC & \textbf{1.00~[0.99--1.00]} & 0.00~[0.00--0.01] & 1.00~[0.99--1.00] & 0.12~[0.10--0.16] & 0.05~[0.03--0.07] \\
CK & 1.00~[0.99--1.00] & 0.00~[0.00--0.01] & 1.00~[0.99--1.00] & 1.00~[0.99--1.00] & 1.00~[0.99--1.00] \\
KC & \textbf{1.00~[0.99--1.00]} & 0.00~[0.00--0.01] & 1.00~[0.99--1.00] & 0.06~[0.04--0.08] & 0.05~[0.03--0.07] \\
KK & 1.00~[0.99--1.00] & 1.00~[0.99--1.00] & 1.00~[0.99--1.00] & 1.00~[0.99--1.00] & 1.00~[0.99--1.00] \\
\bottomrule
\end{tabular}}
\caption{Mixed-type data: size-corrected power on a spread alternative across the four $(X,Y)$
type pairings (C continuous, K categorical; threshold set to $5\%$ on the matched per-type null;
Wilson 95\% CI). GFCM and BLITZ keep full power everywhere. FFCI's null size is $1.0$ whenever a
continuous variable is present, so size correction zeros its power (CC at $n=2000$; CC, CK, KC at
$n=10^5$). The mixed tests based on covariance are blind to the spread edges with continuous response
CC and KC.}
\label{tab:mixed}
\end{table}

\subsection{Speed and scalability}
\label{sec:speed}
Speed is not where GFCM's contribution lies; we report it to confirm the test runs at regression cost. Figure~\ref{fig:speed} reports per-call wall time across conditioning depth and sample size, with GFCM run from its Python reference implementation and analytic $\chi^2$ calibration. Timings come from the reproduction engine, executed in an isolated container on an otherwise idle machine with all cores available to every test alike (AMD EPYC 7502, 32 cores / 64 threads at 2.5\,GHz, Ubuntu 24.04 LTS); each test object is built once and a single per-query call is timed after warm-up. Two patterns hold. First, the random feature tests RCoT and FFCI are insensitive to $|S|$ (a fixed feature count, no growing nuisance), so they are the fastest at deep conditioning; but that flatness is the same one that leaves them unable to target the tail (Section~\ref{sec:experiments}), so their deep-$|S|$ speed is bought with the calibration and tail power they lose (Table~\ref{tab:calib-n1e5} and Figure~\ref{fig:tail-depth}). Second, among the tests that fit a genuine nuisance and therefore grow with $|S|$ (GFCM and BLITZ), the two are comparable: BLITZ is somewhat faster at moderate $n$, while GFCM scales better and overtakes it at large $n$ (at $n=10^5$, $|S|=20$, $22$\,s versus $55$\,s). GFCM therefore runs at the same regression cost as the only other calibrated test, while being sensitive to tail dependence. The partial copula test and the kernel and operator tests are moderate-$n$ only and omitted.

\begin{figure}[tbp]\centering
\includegraphics[width=\textwidth]{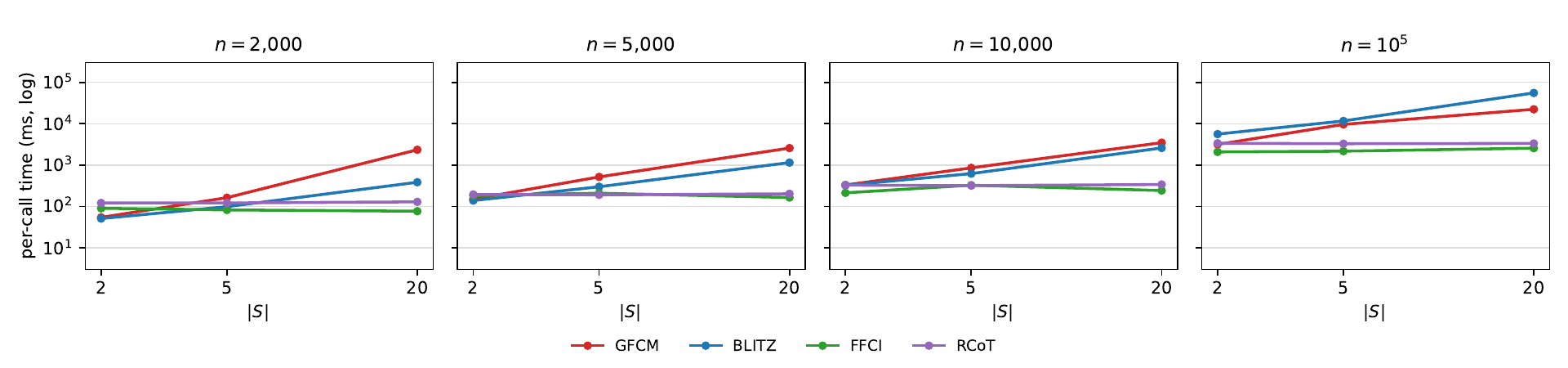}
\caption{Per-call wall time (log scale, ms) for one CI test versus conditioning depth $|S|$, at four sample sizes, measured with the reproduction engine in an isolated container on an otherwise idle machine with all cores available to every test alike (AMD EPYC 7502, 32 cores / 64 threads, Ubuntu 24.04 LTS). GFCM runs from its Python reference implementation with analytic $\chi^2$ calibration; BLITZ and RCoT from their R packages and FFCI from Tetrad, each timed in its native implementation under the same conditions. The random feature tests RCoT and FFCI are flat in $|S|$ (a fixed feature count), while GFCM and BLITZ grow with $|S|$; the two share the same regression-cost tier.
}
\label{fig:speed}
\end{figure}

\subsection{Validation on real data with semi-synthetic injection}
\label{sec:realdata}
The advantage on the skew edge Section~\ref{sec:tailedge} was shown on a fully synthetic generator.
To check that it survives on real heavy-tailed data, we run a semi-synthetic injection.
The noise pool is the set of daily innovations of the EuStockMarkets indices (DAX, SMI, CAC,
FTSE): we take log-returns, standardize each series, and pool them, giving $n=7436$ innovations
with excess kurtosis $4.26$, far above Gaussian. Into this real, heavy-tailed pool we inject an edge in the conditional skew with known ground truth, the centerpiece construction of Section~\ref{sec:tailedge}: $X$ drives the conditional skew of $Y$ while the conditional mean and variance stay flat in $X$, and the conditioning set has $|S|=3$. The null
draws the same pool with no edge. We use $n=2000$ and $R=200$ replications, and report both the
empirical size and the size-corrected power so that the liberal tests are credited only for power
above their own inflated false positive rate (Table~\ref{tab:inject}).

GFCM recovers the injected edge at full power ($1.00$) while holding its size at nominal ($0.045$). The fast residualisation tests detect it too but lower: FFCI \citep{ci105} is the nearest competitor at $0.81$, though only at a mildly liberal size ($0.11$) on the heavy-tailed pool, and RCoT \citep{ci9} ($0.65$) and BLITZ \citep{ci100} ($0.64$) trail. The covariance and copula family is blind, as on the synthetic generator: boosted GCM \citep{ci32} $0.07$, the partial copula test of \citet{ci90} at chance ($0.04$, structurally insensitive to skew), and linear Fisher-$z$ $0.07$. The synthetic result therefore carries over to genuinely heavy-tailed real innovations.

\begin{table}[t]\centering\small
\setlength{\tabcolsep}{6pt}
\begin{tabular}{lcc}
\toprule
test & size & size-corrected power \\
\midrule
\textbf{GFCM} & 0.045~[0.024--0.083] & \textbf{1.000~[0.981--1.000]} \\
FFCI & 0.110~[0.074--0.161] & 0.810~[0.750--0.858] \\
RCoT & 0.055~[0.031--0.096] & 0.645~[0.577--0.708] \\
BLITZ & 0.040~[0.020--0.077] & 0.640~[0.571--0.703] \\
boosted GCM & 0.070~[0.042--0.114] & 0.070~[0.042--0.114] \\
PartCopula & 0.100~[0.066--0.149] & 0.035~[0.017--0.070] \\
Fisher-$z$ & 0.030~[0.014--0.064] & 0.065~[0.038--0.108] \\
\bottomrule
\end{tabular}
\caption{Semi-synthetic injection on real EuStockMarkets innovations (pooled standardized
log-returns, excess kurtosis $4.26$; $n=2000$, $|S|=3$, $R=200$; Wilson 95\% CI). $X$ drives the
conditional skew of $Y$ with conditional mean and variance held flat. Power is size-corrected
(each test re-thresholded to exact $5\%$ on its matched null). GFCM recovers the edge at full
power while holding nominal size; the covariance and copula family is blind to the conditional
skew.}
\label{tab:inject}
\end{table}

The financial injection controls the marginal but reuses a synthetic conditioning structure; we complement it with one built on the real causal structure of the Causal Chambers~\cite{app20}, a physical system whose ground-truth graph is fixed by the device's construction and validated by interventions. We use its light tunnel. On its $9$ edges we change only the edge type: each edge is engineered into a scale edge that preserves the mean, keeping the real sources and each sensor's own noise. The skeleton is unchanged and the sensors are conditionally independent given the sources but strongly correlated in the raw signal. With the real edges the covariance tests recover most structure through the mean, but once the edges are scale edges the covariance and mean family recovers none (recall $0.00$ for Fisher-$z$, linear and boosted GCM), while GFCM recovers $89\%$ (Table~\ref{tab:cc-inject}). The scale dependence is engineered; RCoT and FFCI also recover the injected structure ($0.78$ each), as expected of any flexible test on a scale edge, so this experiment isolates the covariance family's blindness rather than a capability unique to GFCM. GFCM's separation from the flexible tests is a property of the skew edge, shown on the financial injection above; extending it to an uncontrolled tail advantage on naturally occurring data remains future work.

\begin{table}[ht]\centering
\begin{tabular}{lcccc}
\toprule
 & \multicolumn{2}{c}{real CC} & \multicolumn{2}{c}{scale-injected} \\
\cmidrule(lr){2-3}\cmidrule(lr){4-5}
Test & SHD & recall & SHD & recall \\
\midrule
Fisher's $Z$ & 7 & 0.78 & 13 & \textbf{0.00} \\
GCM & 7 & 0.78 & 12 & \textbf{0.00} \\
boosted GCM & 6 & 0.56 & 12 & \textbf{0.00} \\
RCoT & 5 & 0.89 & 4 & 0.78 \\
FFCI & 6 & 1.00 & 6 & 0.78 \\
BLITZ & 8 & 0.33 & 7 & 0.44 \\
PartCopula & 10 & 0.33 & 8 & 0.33 \\
\textbf{GFCM} & 5 & 0.89 & 5 & \textbf{0.89} \\
\bottomrule
\end{tabular}
\caption{Causal Chamber light tunnel, real versus scale-injected ($n=2000$). Engineering the real edges into scale edges that preserve the mean makes the covariance and mean family blind (recall $0.00$ for Fisher-$z$, linear and boosted GCM), while GFCM recovers $89$ (RCoT and FFCI also recover, $0.78$; BLITZ and the partial copula test are weak on this real data throughout}
\label{tab:cc-inject}
\end{table}

\subsection{GFCM inside PC}
\label{sec:pc}
We embed each test in full PC discovery on random typed DAGs of 8 vertices (edge probability 0.4, linear, non-monotone, scale, and tail edges) and score the learned skeleton against the truth, with GFCM running the two fixes of Section~\ref{sec:fixes} (both orientations for F1, rank-transformed conditioning for F2). The outcome is a crossover in sample size (Figure~\ref{fig:pc}). At $n=2000$, FFCI has the lowest SHD ($3.70$) and GFCM is a close second ($4.12$), both controlling false edges (precision $0.96$--$0.97$). At $n=10^5$ this reverses: GFCM attains SHD $0.61$, half the nearest competitor (BLITZ $1.22$), recovering almost every edge including the location-scale and tail edges the covariance family deletes (recall $0.99$). FFCI loses false-positive control at scale, its precision falling from $0.96$ to $0.84$ as it over-keeps edges (it can no longer find deep separating sets), making it the worst of the four. GFCM is thus the lowest-SHD test among those that remain size-valid at scale, and its margin over BLITZ and RCoT widens with $n$.

\begin{figure}[t]\centering
\includegraphics[width=\textwidth]{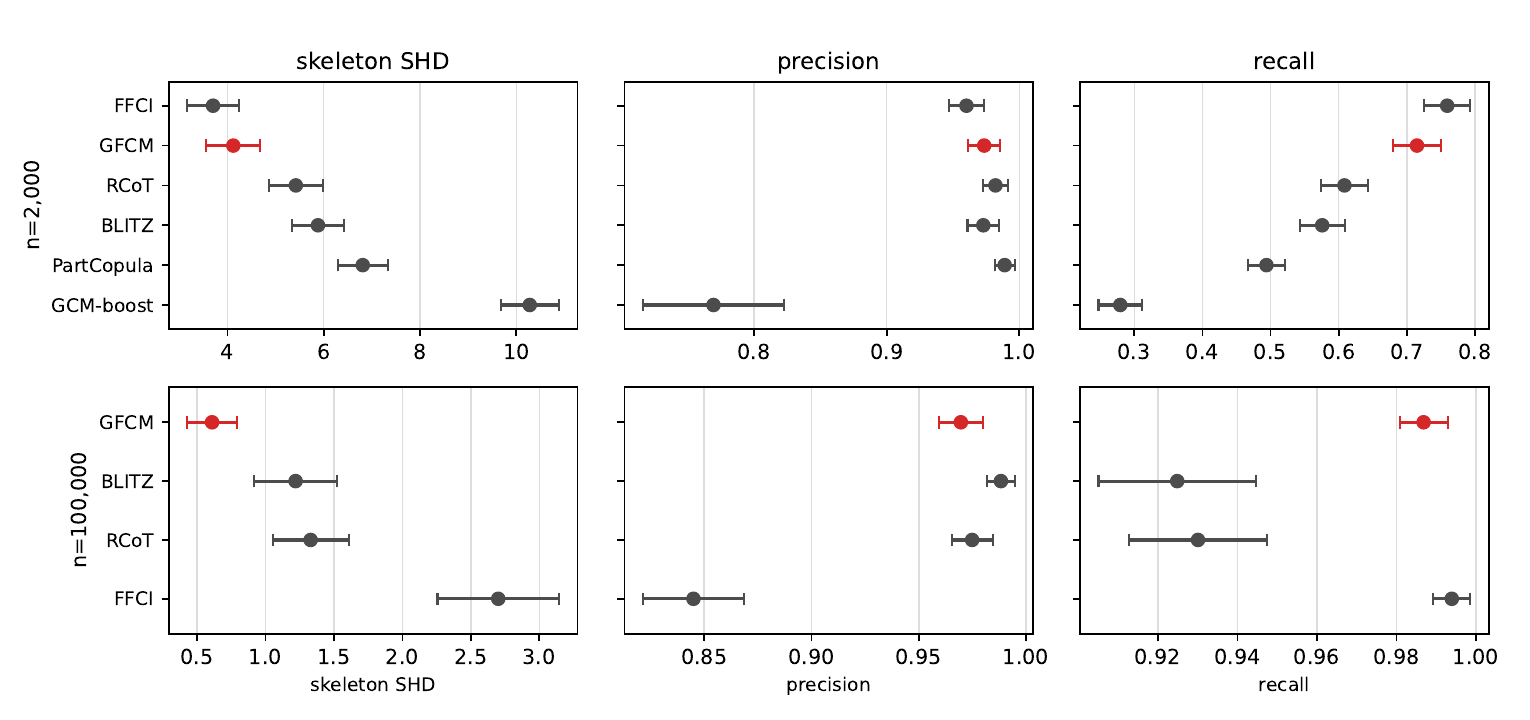}
\caption{PC discovery on random typed DAGs (8 vertices, edge probability 0.4): skeleton SHD (lower is better), precision, and recall
over $100$ replicates (normal $95\%$ CIs; GFCM in red, tests ordered by SHD). A sample-size
crossover: FFCI leads on SHD at $n=2000$ but loses precision control at $n=10^5$ (its deep
separating sets fail), while GFCM attains the lowest SHD at scale, half the nearest competitor, at
high precision and near-perfect recall.}
\label{fig:pc}
\end{figure}

\section{Related Work}
\label{sec:related}

Our contribution is the combination rather than any single ingredient: quantile
regression, identification functions, and the GCM template all predate this work. We
are not aware, however, of an existing CI test that is simultaneously
constraint-based, nonparametric, and mixed-type, detects location-scale/tail and
non-monotone mean dependence at regression cost without expert knowledge or an assumption of functional model, and is hardened to run inside PC; that conjunction is the gap this paper targets.

The closest constraint-based work, ParCorr-WLS~\cite{ci36},
corrects heteroscedastic noise to protect a mean (partial correlation) test,
needs expert knowledge of the heteroscedasticity structure, and, using symmetric
mean-mean residuals, cannot detect pure scale edges and does not encounter our two
failure modes. The functional model line (LSNM/LOCI~\cite{algo102},
QCCD~\cite{algo103}) is bivariate cause and effect inference; its asymmetry is
identifiability of causal direction, not a CI test orientation, and it has no
conditioning sets. Quantile graphical models~\cite{algo91} produce an
undirected quantile CI graph. Quantile/copula CI tests~\cite{ci90} target the
full conditional distribution on continuous data and are not used for discovery.
The GCM family~\cite{ci32,ci50,ci52} is the covariance baseline we generalize.

On the
distributional side, the kernel and operator tests (KCI~\cite{ci24}, GKCM~\cite{ci60},
RCoT~\cite{ci9}) and especially SGCM~\cite{ci67}, which targets the full conditional
cross-covariance operator through a data-adaptive spectral feature expansion, with a
complete uniform asymptotic theory on general spaces, detect strictly more than
we do, mixed-type included. We do not compete on detection breadth. We differ on the
three axes they leave open: SGCM's operator norm is undirected and so admits no dial to target a tail level (our characterisation of detection class); its $O(n^2)$ multiplier
bootstrap is heavier than our $O(n)$ statistic and its study stops at $n\le600$ (a faster
variant is not ruled out, but as published it does not reach discovery scale); and,
like all of them, it is a standalone test, not hardened for the F1/F2 asymmetries that
arise inside PC. The recent fast residualisation test BLITZ~\cite{ci100}, like RCoT, shares our speed tier
and is well calibrated, so we do not differentiate on cost against it; both remain
not natively mixed-type (RCoT is continuous only; BLITZ only through a jittered integer coding
that is not invariant to relabeling, though it matches our power on the mixed benchmark), offer no dial
for a chosen tail level, and are standalone, leaving the targeting and integration gaps inside PC we fill,
plus native mixed-type handling. Our claim is the configuration,
the discovery integration, and the conjunction, not the operator or raw speed.

A separate and growing line builds causal and graphical models specifically for extremes:
extremal graphical models~\cite{algo95}, causal discovery in heavy-tailed models~\cite{algo96},
the extremes of structural causal models~\cite{algo97,algo98}, and applied networks of extremal dependence~\cite{app41,app30}. These rest on extreme value machinery (multivariate Pareto,
H\"usler--Reiss, regular variation), model continuous tail-only data in modest dimension, and
are not general conditional independence tests for constraint-based discovery. We are
complementary: rather than modelling the extremes, we bring location-scale and tail sensitivity
into a standard mixed-type CI test usable inside PC.

\section{Discussion}
\label{sec:discussion}
The standard CI tests for causal discovery are covariance tests, structurally blind to
non-monotone mean and location-scale dependence. Distributional tests recover that
dependence, and the strongest of them, the kernel and operator tests, detect more than
we do. Our point is not detection but deployment: a configurable feature covariance
test that targets the dependence of interest, runs at regression cost on mixed-type data
where the operator tests wall out, and is hardened against the two failure modes (F1, F2)
that arise only once a test sensitive beyond covariance is placed inside PC. The characterisation of detection class says what a configuration sees and misses; the validity analysis
pins the binding condition on the nuisance rate; and on random DAGs at scale PC with this test has the
lowest skeleton SHD among size-valid tests, while running orders of magnitude faster than the kernel and operator
tests.

\paragraph{Limitations.}
\begin{itemize}
\item The advantage scales with how the dependence is dominated by the tail or the scale: on pure mean structure GFCM ties the mean family, and on mixed edges that mix mean and scale the mean tests recover part of it, narrowing though not closing the gap.
\item Detection beyond covariance is shared with distributional tests such as KCI; our advantage over them is scalability and mixed-type support, not detection.
\item The tail/scale advantage is shown on simulated and semi-synthetic data (financial innovations and the scale-injected Causal Chamber light tunnel, Section~\ref{sec:realdata}), not on naturally tail-heavy real data: no benchmark combines real data, trustworthy causal ground truth, and effects located in the tail, since the tail and variance domains (finance, single-cell regulation) have proxy or contested ground truth while systems with clean ground truth are not tail-heavy.
\item The ACAT combination is asymptotically valid but mildly anticonservative in the extreme tail, which compounds with GFCM's higher per-test power to inflate false positives on dense graphs; the fix is error control at the graph level (FDR or $e$-values) in the skeleton search, an important future direction.
\item We prove validity, the detection class, and a $\Phi$-faithfulness consistency result (Theorem~\ref{thm:phicpdag}): PC$+$GFCM recovers the CPDAG of the graph whose independence model is the set's detection class, and is conservatively sparser otherwise (Corollary~\ref{cor:phiconservative}). Sharper identification of the quantile CPDAG across the level grid remains future work.
\item Per-call cost is characterized only to $|S|=20$, where GFCM is competitive among the calibrated tests in the shallow region PC operates in (Section~\ref{sec:speed}); its behavior at $|S|$ well beyond $20$ and total discovery cost on much larger graphs are future work, and GFCM is the most accurate but not the fastest test overall (the random feature tests are faster at deep conditioning, at the cost of the tail sensitivity they forfeit).
\item Validity and the regression-cost speed rest on the nuisance being additive in $Z$ up to the default degree-3 interaction block, and, for the spline default, on not simultaneously having a heavy-tailed conditioner and a non-polynomial conditional mean; both boundaries can be dialed at a speed cost (the gradient-boosting backend, ${\sim}40\times$ slower, Section~\ref{sec:nuisance}, Table~\ref{tab:interaction}), not failures of the test target.
\item The set extends beyond the default to scale-scale ($c_Xc_Y$) and tail-tail ($r_\tau(X)\,r_\sigma(Y)$) products that target conditional co-dispersion and joint-tail dependence and inherit the validity theory, but their empirical study and the data generating processes it requires are left to future work.
\end{itemize}


\section*{Reproducibility statement}
The GFCM test and the full reproduction pipeline, including the exact simulation configurations, seeds, aggregation scripts, and the pinned container image behind every reported number, are available at github.com/averinpa/gfcm. All reported numbers use the canonical GFCM configuration.

\bibliographystyle{tmlr}
\bibliography{references}

\appendix
\section*{Appendix}
\section{Validity}
\label{app:proofs}

GFCM combines components of the form $\psi=\phi_X\,\phi_Y$, a product of a feature of $X$
and a feature of $Y$ where at least one factor is a residual that has mean zero given $Z$. The
continuous features are the mean residual $e_A=A-m_A(Z)$, the centered scale
$c_A=|e_A|-\E[|e_A|\given Z]$ (default $|e|$, or the squared form $e_A^2-\E[e_A^2\given Z]$;
mean zero given $Z$ by construction, and the centering argument below is identical for either), and the
quantile indicator residual $r_\tau=\tau-\mathbf 1\{A\le q_\tau(Z)\}$; for a categorical
$A$ the feature is the centered one-hot residual. The configured null is
$X\indep Y\given Z$ (for the quantile feature, $\E[r_\tau\given X,Z]=0$).

\begin{lemma}[Mean zero]
\label{lem:meanzero}
Let $\phi_X\in\mathcal E_{X,Z}=\{f\in L^2_{X,Z}:\E[f\given Z]=0\}$ and $\phi_Y\in L^2_{Y,Z}$
(at least one factor conditionally centered, taken as $\phi_X$). Under $X\indep Y\given Z$ the
factors are conditionally independent given $Z$, so
$\E[\psi\given Z]=\E[\phi_X\given Z]\,\E[\phi_Y\given Z]=0$ and $\E[\psi]=0$, the product
$\psi=\phi_X\phi_Y$ being integrable by Cauchy--Schwarz.
\end{lemma}

This is the forward direction of Daudin's characterization of conditional independence
(\citealp{ci107}; condition~(iii) of \citealp{ci24}, Lemma~2; the moment identity is Eq.~(1)
of \citealp{ci32}). Its converse, that vanishing of $\E[\phi_X\phi_Y]$ over all of
$\mathcal E_{X,Z}\times L^2_{Y,Z}$ is also sufficient for conditional independence, is what
makes the set's detection class (Section~\ref{sec:detection}) the right object; we use only
the forward direction here. The GCM is the mean residual case $\phi_X=e_X$.

\begin{lemma}[Neyman orthogonality]
\label{lem:orth}
At the truth, the moment of any product of two conditionally centered features is
Neyman orthogonal to every nuisance entering either factor (\citealp{ci25}; the moment is the
GCM product of \citealp{ci32}): under the null the factors are conditionally independent given
$Z$, so perturbing a nuisance in one factor has a first-order effect annihilated by the
centered partner, $\E[\phi\given Z]=0$ (symmetrically for the other). This covers the set's
pairs and the higher-order products $c_Xc_Y$, $r_\tau(X)\,r_\sigma(Y)$
(Section~\ref{sec:discussion}).
(i) Mean-quantile $e_X r_\tau(Y)$: the $m_X$-derivative vanishes since $\E[r_\tau\given Z]=0$
at the true quantile \citep{koenker1978}; perturbing $m_Y$, $s_Y$, or $q_\tau$ moves the
location-scale threshold $Q_\tau=m_Y(Z)+s_Y(Z)\,q_\tau$ by an amount that is a function of $Z$
times the density $f_{Y\given Z}(Q_\tau\given Z)$; since it depends only on $Z$ it factors out of
the conditional expectation and is annihilated by $\E[e_X\given Z]=0$.
(ii) Centered transform $c_X e_Y$: the $m_Y$-derivative is killed by $\E[c_X\given Z]=0$ and
the derivative in the variance nuisance by $\E[e_Y\given Z]=0$; the centered one-hot residual is as in (i).
\end{lemma}

Centering is essential in (ii): the uncentered $e_X^2 e_Y$ has a nonzero $m_Y$-derivative
$-\E[\Delta_Y\,\Var(X\given Z)]$, a first-order bias; subtracting $\E[e_X^2\given Z]$ (a cheap
auxiliary regression) restores double robustness, the same higher-moment centering used in
the spectral expansion of \citet{ci67}. The density term in (i) is the ingredient the mean
GCM lacks.

\begin{theorem}[Asymptotic level]
\label{thm:validity}
Assume:
\begin{itemize}
\item[\textup{(A1)}] all nuisances ($\hat m_X,\hat m_Y$, the scale function
$\widehat{\E}[|e|\given Z]$, the variance $\widehat{\E}[e^2\given Z]$ for the squared variant,
and $\hat q_\tau$) are additive in the coordinates of $Z$, $m(z)=\sum_{j=1}^{d}g_j(z_j)$ with
each $g_j$ H\"older of smoothness $s>\tfrac12$ on the rank-transformed scale, cross-fitted on
the additive cubic spline basis whose per-coordinate knot count grows as
$K(n)\asymp n^{1/(2s+1)}$ (default $n^{1/5}$, matching $s=2$);
\item[\textup{(A2)}] each nuisance converges at $o_P(n^{-1/4})$ in $L_2$, with cross-products
$o_P(n^{-1/2})$, uniformly in $z$ (equivalently, the product of the two regression
mean squared errors is $o_P(n^{-1})$, the GCM rate condition of \citealp{ci32}, Theorem~6 and
Remark~7; \citealp{ci25} for the general DML form);
\item[\textup{(A3)}] $f_{Y\given Z}(\cdot\given z)$ is bounded and Lipschitz near
$q_\tau(z)$, and $\E[X^4\given Z\!=\!z],\E[Y^4\given Z\!=\!z]$ are bounded, uniformly in $z$;
\item[\textup{(A4)}] the per-block score covariance
$R=\lim_n n^{-1}\sum_i\E[\psi_i\psi_i^\top\given Z_i]$ exists, positive definite for
signal-bearing blocks (a structurally null block, Proposition~\ref{prop:f1}, is admitted with
degenerate $R$: its sign-flip $p$-value stays valid without positive definiteness, the
exactness clause of Lemma~\ref{lem:signflip}, and the Cauchy combination draws its rejection
from a non-degenerate block); and the conditional
scale is bounded below, $\hat s=\widehat\E[|e|\given Z]\ge\underline s>0$, so the
standardization $u=e/\hat s$ is well defined.
\end{itemize}
Then under $X\indep Y\given Z$ each centered, orthogonal component
(Lemma~\ref{lem:orth}) is asymptotically normal with mean zero, its sign-flip statistic
is calibrated (Lemma~\ref{lem:signflip}), and the Cauchy combination of the component
$p$-values has asymptotic level $\alpha$.
\end{theorem}

\paragraph{Remark.} Under (A1) the rate is dimension-free: the additive design attains the
one-dimensional rate regardless of $d$ \citep{ci110}, so the size conditioning set $d=|S|$ enters
the $L_2$ error $O(\sqrt d\,n^{-s/(2s+1)})$ through the constant $\sqrt d$, not the exponent,
unlike the general $d$-dimensional rate $n^{-s/(2s+d)}$ \citep{ci109} where $d$ enters the
exponent (the curse of dimensionality). Thus (A2) holds for every fixed $d$, so the level guarantee is agnostic to dimension and carries
to any size of conditioning set $|S|$. The binding condition is additivity, not dimension: a fixed basis
(parametric or fixed-$K$ spline) has non-vanishing bias, so the orthogonal remainder stays
first-order and size is eventually distorted; GFCM therefore ships the growing
$K(n)\approx10\,n^{1/5}$ that (A1) describes, and the nonlinear mean null of
Section~\ref{sec:nuisance} confirms the additive spline holds nominal level where a fixed
low-degree basis distorts size. The default augments this additive backbone
with two blocks of the same growing basis, each leaving every score with its centered factor
so Lemma~\ref{lem:orth} and the (A2) rate are unchanged: an interaction block of degree 3
($O(|S|^2)$ columns, active to $|S|\le20$) and, at $n\ge n_0$, a single-index PROJ block
(Section~\ref{sec:gfcm}).

\begin{proof}
Let $\hat\psi_{\tau,i}=(X_i-\hat m_X^{(-k_i)}(Z_i))(\tau-\mathbf1\{Y_i\le\hat
q_\tau^{(-k_i)}(Z_i)\})$ be the cross-fitted score, where $\hat m_X^{(-k_i)},\hat
q_\tau^{(-k_i)}$ are fitted on the folds excluding the one containing $i$, and let
$\psi_{\tau,i}$ be its oracle (true nuisance) version. Decompose
\[
\tfrac1{\sqrt n}\textstyle\sum_i\hat\psi_{\tau,i}
=\underbrace{\tfrac1{\sqrt n}\textstyle\sum_i\psi_{\tau,i}}_{\mathrm{(I)}}
+\underbrace{\tfrac1{\sqrt n}\textstyle\sum_i\big(\Delta_i-\E[\Delta_i\mid\mathcal D^{(-k_i)}]\big)}_{\mathrm{(II)}}
+\underbrace{\tfrac1{\sqrt n}\textstyle\sum_i\E[\Delta_i\mid\mathcal D^{(-k_i)}]}_{\mathrm{(III)}},
\]
with $\Delta_i=\hat\psi_{\tau,i}-\psi_{\tau,i}$ and $\mathcal D^{(-k_i)}$ the data used
to fit $i$'s nuisances.

(II), empirical process, no Donsker condition. Conditional on $\mathcal
D^{(-k_i)}$, the summands are independent, mean zero, with variance bounded by
$\E[\Delta_i^2\mid\mathcal D^{(-k_i)}]\to_P0$ under (A2). Cross-fitting makes the
nuisance independent of the evaluation fold, so $\Var(\mathrm{(II)}\mid\mathcal
D^{(-k_i)})=o_P(1)$ and $\mathrm{(II)}=o_P(1)$ by Chebyshev. No
Donsker condition on the nuisance class is needed, this is
what sample splitting buys.

(III), bias, orthogonality controls the non-smooth indicator. By
Lemma~\ref{lem:orth} the pathwise derivative of $\E\psi_\tau$ in each nuisance
vanishes, so (III) is a second-order remainder. The mean residual direction
contributes $\E[(\hat m_X-m_X)\,\E[r_\tau\mid Z]]=0$ exactly. The indicator enters
only through the conditional quantile, which in the location-scale model is
$\widehat Q_\tau(Y\given Z)=\hat m_Y(Z)+\hat s_Y(Z)\,\hat q_\tau$, so its error aggregates
the mean, scale, and errors in standardized level,
\[
\|\widehat Q_\tau-Q_\tau\|\le\|\hat m_Y-m_Y\|+|q_\tau|\,\|\hat s_Y-s_Y\|
+\bar s\,|\hat q_\tau-q_\tau|=o_P(n^{-1/4})
\]
by (A2) (the scale enters because the quantile reuses the scale fit rather
than running a separate quantile regression; $\hat s_Y\ge\underline s$ by (A4)). A
one-term expansion of $\Pr(Y\le\widehat Q_\tau\mid Z)$ about $Q_\tau$ with bounded
Lipschitz density (A3) gives the linear term $f_{Y\mid Z}(Q_\tau\mid Z)(\widehat
Q_\tau-Q_\tau)$, annihilated by $\E[e\mid Z]=0$, plus an $O(\|\widehat
Q_\tau-Q_\tau\|^2)$ remainder. Hence $|\mathrm{(III)}|\le\sqrt n\big(\|\hat
m_X-m_X\|\,\|\widehat Q_\tau-Q_\tau\|+\tfrac L2\|\widehat
Q_\tau-Q_\tau\|^2\big)=o_P(1)$ by (A2). The non-smoothness costs only the assumption of bounded Lipschitz density (A3), not smoothness of the statistic.

(I), the leading term. For each component, $n^{-1/2}\sum_i\psi_i$ averages
i.i.d., conditionally centered scores (Lemma~\ref{lem:meanzero}). The uniform fourth-moment bound (A3)--(A4) gives
a Lyapunov condition, $n^{-1}\sum_i\E|\psi_i|^{2+\delta}=O(1)$ with $\delta=2$, hence the
Lindeberg condition; with the positive definite limit covariance $R$ of (A4) the
Lindeberg--Feller CLT gives $n^{-1/2}\sum_i\psi_i\to_d\mathcal N(0,R)$, non-degenerate, and
$T_n^\top\widehat R^{-1}T_n\to_d\chi^2_k$. The quantile and centered one-hot residuals are bounded
($|r_\tau|\le1$), so only the two moment components invoke (A3); this is why heavy-tailed
or zero-inflated data does not threaten the CLT, unlike the mean GCM whose score inherits
the tail in every component.

Combination. By Lemma~\ref{lem:signflip} each component's sign-flip statistic is
asymptotically calibrated; the Cauchy combination of (possibly dependent) component
$p$-values controls level under arbitrary dependence, so $\Pr(\text{reject})\to\alpha$.
\end{proof}

The single-index block (Section~\ref{sec:gfcm}) restores validity on a non-additive single-index
mean $g(\sum_j Z_j)$, the case the additive design cannot fit ((Table~\ref{tab:nuisance-tradeoff})), within
the same DML rate machinery: an estimated index direction plus a one-dimensional link meets the
rate condition (A2).

\begin{proposition}[Single-index nuisance validity]
\label{prop:singleindex}
Add the assumptions:
\begin{itemize}
\item[\textup{(SI1)}] the conditional means are single-index, $\E[X\mid Z]=g_X(a_X^\top Z)$ and
$\E[Y\mid Z]=g_Y(a_Y^\top Z)$, for unit directions $a_X,a_Y$ and links $g_X,g_Y$ H\"older of
order $s>\tfrac12$ ($s=2$ for the cubic spline) with bounded first derivative;
\item[\textup{(SI2)}] $Z$ is elliptically symmetric with positive definite covariance, so by
Stein's identity the ridge-OLS (average derivative) coefficient is proportional to the index
direction, and its cross-fitted estimate obeys $\|\hat a-a\|=O_P(n^{-1/2})$
\citep{ci111,ci112,ci113};
\item[\textup{(SI3)}] the link is a cross-fitted cubic spline on the estimated index
$\hat a^\top Z$ with $K(n)\asymp n^{1/(2s+1)}$.
\end{itemize}
Then under \textup{(SI1)--(SI3)} and the conditions of Theorem~\ref{thm:validity}, the
nuisance augmented by single-index $\widehat\E[X\mid Z]=\hat g_X(\hat a_X^\top Z)$ (and likewise for
$Y$) satisfies the rate condition \textup{(A2)}, so the level guarantee of
Theorem~\ref{thm:validity} extends to it: GFCM with the single-index block is asymptotically
valid on a mean $g(\sum_j Z_j)$ that the additive spline cannot represent and on which the
additive-only design distorts size (Table~\ref{tab:nuisance-tradeoff}).
\end{proposition}

\begin{proof}[Proof sketch]
Split the nuisance error into a link term and a direction term,
\[
\hat g(\hat a^\top Z)-g(a^\top Z)
=\underbrace{\big[\hat g(\hat a^\top Z)-g(\hat a^\top Z)\big]}_{\text{link}}
+\underbrace{\big[g(\hat a^\top Z)-g(a^\top Z)\big]}_{\text{direction}}.
\]
\textit{Link.} The link is fit on the estimated index $\hat a^\top Z$, not the true $a^\top Z$.
Cross-fitting makes $\hat a$ independent of the evaluation fold, and the direction error
$\|\hat a-a\|=O_P(n^{-1/2})$ is faster than the univariate link rate, so regressing on
$\hat a^\top Z$ rather than $a^\top Z$ does not degrade it: the one-dimensional spline attains the
standard univariate nonparametric rate $\|\hat g-g\|_{L_2}=O_P(n^{-s/(2s+1)})=O_P(n^{-2/5})$ for
$s=2$ (single-index regression with a consistent estimated direction; \citealp{ci113,ci109}).

\textit{Direction.} By the
bounded derivative in (SI1), $|g(\hat a^\top Z)-g(a^\top Z)|\le\|g'\|_\infty\,|(\hat a-a)^\top Z|$,
so in $L_2$ it is $O_P(\|\hat a-a\|)=O_P(n^{-1/2})$ by (SI2).

\textit{Combining,} the nuisance error is
$O_P(n^{-2/5})=o_P(n^{-1/4})$, the individual rate in (A2); the cross-product is
$O_P(n^{-2/5})\cdot O_P(n^{-2/5})=O_P(n^{-4/5})=o_P(n^{-1/2})$, the product rate. The orthogonality,
score CLT, and sign-flip/$\chi^2$ calibration of Theorem~\ref{thm:validity} enter only through (A2)
and are unchanged.
\end{proof}

\paragraph{Remark (the gate is the finite-sample boundary of (SI2)).} The direction rate is
asymptotic: the constant in $\|\hat a-a\|=O_P(n^{-1/2})$ scales inversely with the index
signal-to-noise, so at small $n$ the estimated direction is noisy and the direction term is
non-negligible, the empirically observed small-$n$ liberality. The single-index block is
therefore gated to $n\ge n_0$ (Section~\ref{sec:gfcm}); the gate is where (SI2)'s constant is small
enough, not an additional assumption.

\paragraph{Remark (scope, and why the polynomial block is retained).} (SI1) is a single-index
restriction: it covers single-index non-additive means such as $g(\sum_j Z_j)$ but not general
interactions. A bilinear mean $Z_1Z_2$ is not single-index (its average derivative direction
vanishes by symmetry), which is why the polynomial interaction block is kept
(Table~\ref{tab:nuisance-tradeoff}). Extending (SI2) to such symmetric interactions via
second-moment (SAVE) directions requires a separate condition on eigenvalue condition, left to future work.

\begin{lemma}[Sign-flip calibration]
\label{lem:signflip}
Let $S_n^\circ=T_n^{\circ\top}\widehat R^{-1}T_n^\circ$ with $T_n^\circ=n^{-1/2}\sum_i
\varepsilon_i\hat\psi_i$ and $\varepsilon_i\overset{\mathrm{iid}}\sim\mathrm{Unif}\{\pm1\}$
independent of the data, and $\widehat R=n^{-1}\sum_i\hat\psi_i\hat\psi_i^\top$. Under
(A1)--(A4), conditional on the data, $T_n^\circ\to_d\mathcal N(0,R)$ and
$S_n^\circ\to_d\chi^2_k$ in probability, and the convergence is at Kolmogorov rate
$O_P(n^{-1/2})$: the sign-flip $\alpha$-quantile equals the $\chi^2_{k,1-\alpha}$ quantile
up to $O_P(n^{-1/2})$, so the test of Theorem~\ref{thm:validity} has level
$\alpha+O_P(n^{-1/2})$. If the score is conditionally sign symmetric
($\hat\psi_i\mid Z_i\overset d=-\hat\psi_i\mid Z_i$), the sign-flip test is
finite-sample exact.
\end{lemma}
\begin{proof}
Conditional on the data, $T_n^\circ$ sums independent mean zero terms
$\varepsilon_i\hat\psi_i$ with conditional covariance $\widehat R$, and $\widehat
R=R+O_P(n^{-1/2})$ by the CLT for the i.i.d.\ (cross-fit) terms $\hat\psi_i\hat\psi_i^\top$,
which have finite variance under the uniform $(2+\delta)$-moment bound (A3)--(A4) with
$\delta=2$.

Conditional Lindeberg, explicit. For $\eta>0$ the Lindeberg sum is
$L_n(\eta)=n^{-1}\sum_i\E\!\big[|\hat\psi_i|^2\mathbf1\{|\hat\psi_i|>\eta\sqrt n\}\mid
\text{data}\big]$. On the event the moment factors satisfy
$n^{-1}\sum_i|\hat\psi_i|^{2+\delta}=O_P(1)$ (A3)--(A4), Markov's inequality gives
$|\hat\psi_i|^2\mathbf1\{|\hat\psi_i|>\eta\sqrt n\}\le(\eta\sqrt n)^{-\delta}
|\hat\psi_i|^{2+\delta}$, so $L_n(\eta)\le\eta^{-\delta}n^{-\delta/2}\,
n^{-1}\sum_i|\hat\psi_i|^{2+\delta}=O_P(n^{-\delta/2})\to0$ (for the bounded quantile and
one-hot components $|\hat\psi_i|\le1$ makes the indicator vanish for $n>\eta^{-2}$
outright). The conditional Lindeberg--Feller CLT then gives
$\widehat R^{-1/2}T_n^\circ\to_d\mathcal N(0,I_k)$, matching the unconditional limit of
$T_n$, and the continuous map $\|\cdot\|^2$ gives $S_n^\circ\to_d\chi^2_k$.

Rate. For the quadratic form, the multivariate Berry--Esseen bound over Euclidean
balls~\citep{bentkus2005} gives, conditional on the data,
$\sup_{r}\big|\Pr(\|\widehat R^{-1/2}T_n^\circ\|\le r\mid\text{data})-\Pr(\|N\|\le r)\big|
\le C\,k^{1/4}\,n^{-1/2}\,n^{-1}\!\sum_i\E\|\widehat R^{-1/2}\hat\psi_i\|^3=O_P(n^{-1/2})$,
using the uniform third-moment bound and $\widehat R\succeq c>0$ (A4). Adding the
covariance error $\widehat R-R=O_P(n^{-1/2})$, which perturbs the quadratic form by
$O_P(n^{-1/2})$, the conditional law of $S_n^\circ$ is within $O_P(n^{-1/2})$ of $\chi^2_k$
in Kolmogorov distance; inverting at the $1-\alpha$ level gives the stated quantile and
level rates.

Exactness. Under conditional sign-symmetry the law of $(\varepsilon_i\hat\psi_i)_i$
equals that of $(\hat\psi_i)_i$, so the sign-flip reference distribution is exact at every
$n$, with no appeal to the CLT or its rate.
\end{proof}

\begin{corollary}[Validity-invariance]
\label{cor:validinv}
The same level guarantee applies to any feature of $Y$ whose centered residual $r$
satisfies $\E[r\given Z]=0$ at the true nuisance and the moment conditions of
Theorem~\ref{thm:validity}. The quantile indicator and centered one-hot
residuals are bounded, so they meet these automatically; the identity (mean GCM)
residual additionally requires $Y$ to have finite conditional variance, which is
exactly the condition the bounded quantile residual removes. Within this class the
user's per-variable choice changes only the alternative against which the test has
power (the non-degeneracy direction A4), never the level.
\end{corollary}

\section{Structural directional asymmetry (F1)}
\label{app:f1}
The directional statistic uses one variable's mean residual and the other's quantile
residual, so it is not symmetric in $(X,Y)$. The next proposition shows that for a
pure location-scale edge one of the two orientations has zero population
signal at every level, independently of the nuisance estimator, the conditioning
design, and the quantile set. The blindness is therefore a property of the statistic,
not of the nuisance estimator (the concern of F2).

\begin{proposition}[Directional blindness to pure-scale edges]
\label{prop:f1}
Suppose $\E[Y\mid X,Z]=\E[Y\mid Z]$: $X$ alters the conditional distribution of $Y$
given $Z$ but not its conditional mean (a pure location-scale edge). Let
$e_Y=Y-\E[Y\mid Z]$ and $r_\tau(X)=\tau-\mathbf1\{X\le Q_\tau(X\mid Z)\}$. Then the
$Y$-mean\,/\,$X$-quantile orientation carries no signal at any level,
\[
\beta_{Y\to X}(\tau):=\E[e_Y\,r_\tau(X)]=0\qquad\text{for all }\tau\in(0,1),
\]
so used alone it is consistent only against its own level and PC deletes the edge at
the marginal test. The complementary orientation $\beta_{X\to Y}(\tau)=\E[e_X\,
r_\tau(Y)]$ is nonzero whenever $X$ moves a non-median quantile of $Y$
(Proposition~\ref{thm:detection}); hence the symmetric two-orientation test is consistent
against every pure scale alternative.
\end{proposition}
\begin{proof}
Since $r_\tau(X)$ is a function of $(X,Z)$, $\E[e_Y\,r_\tau(X)\mid Z]=\E[\,\E[e_Y\mid
X,Z]\,r_\tau(X)\mid Z]$. The hypothesis gives $\E[e_Y\mid X,Z]=\E[Y\mid X,Z]-\E[Y\mid
Z]=0$, so the inner expectation is $0$ and $\beta_{Y\to X}(\tau)=\E_Z[0]=0$ for every
$\tau$. The argument uses only the mean independence $\E[Y\mid X,Z]=\E[Y\mid Z]$: it
holds for any conditioning design, any consistent nuisance, and any $\tau$, so the
asymmetry is structural.
\end{proof}

\paragraph{Reconciliation with Theorem~\ref{thm:validity}.} A structurally null
orientation is a degenerate score direction, which is exactly the case assumption
(A4) admits per block rather than excludes globally. The null orientation occupies its
own block, and that block's sign-flip $p$-value (Lemma~\ref{lem:signflip}) is valid even
when its score covariance is rank-deficient, because the sign-flip bootstraps the
empirical law and does not invoke the non-degenerate $\chi^2_k$ limit; the Cauchy
combination then rejects through the complementary, non-degenerate orientation. So
Proposition~\ref{prop:f1} and Theorem~\ref{thm:validity} are consistent on the pure scale
edge: F1 makes one block null, and the test stays valid. We keep the asymmetric
mean-quantile construction despite F1 because it targets a directed
alternative, $X$'s level driving $Y$'s distribution, that a symmetric quantile-quantile
product does not (Section~\ref{sec:ablation}).

\section{Detection class}
\label{app:detection}

With $g_\tau(X,Z)=\Pr(Y\le Q_\tau(Y\given Z)\given X,Z)-\tau$ (so
$\E[g_\tau\given Z]=0$), the quantile effect curve is
$\beta(\tau)=\E[e\,r_\tau]=-\E[e\,g_\tau]$, and $\beta\equiv0$ under the null.

\begin{proof}[Proof of Proposition~\ref{thm:detection}]
By Theorem~\ref{thm:validity} each component score has $\E[\psi_{\tau_i}]=\beta(\tau_i)$ and
a finite limiting variance $\sigma_{\tau_i}^2$, so the studentized statistic
$T_{\tau_i}=\sqrt n\,\widehat\E[\psi_{\tau_i}]/\widehat\sigma_{\tau_i}$ concentrates around
$\sqrt n\,\beta(\tau_i)/\sigma_{\tau_i}$: it is $O_P(1)$ when $\beta(\tau_i)=0$ and diverges
at rate $\sqrt n$ when $\beta(\tau_i)\neq0$. The block's sign-flip Wald statistic
$S_n=T_n^\top\widehat R^{-1}T_n$ therefore diverges iff some component does, its $p$-value
$\to0$; and the ACAT combination rejects iff some block's $p$-value $\to0$, i.e.\ iff
$\beta_{\phi_X\phi_Y}\neq0$ for some admissible pair. If all $\beta_{\phi_X\phi_Y}=0$ every
score has mean zero to first order (Lemma~\ref{lem:orth}) and the combined test keeps its null
level.

(i) Location. If $X$ shifts the conditional distribution,
$g_\tau\approx -f_{Y\given Z}(Q_\tau)\,\delta(X)$, so
$\beta(\tau)\approx f_{Y\given Z}(Q_\tau)\,\E[e\,\delta(X)]$ is single-signed in
$\tau$; detected at any single level.

(ii) Scale. With $Y=m(Z)+s(X,Z)\varepsilon$, $\varepsilon$ symmetric with
median $0$: $Q_{0.5}(Y\given X,Z)=m(Z)$ is independent of $X$, so $g_{0.5}=0$ and
$\beta(0.5)=0$; for $\tau>0.5$, $g_\tau$ depends on $s(X,Z)$ and correlates with
$e$, so $\beta(\tau)\neq0$, with $\beta(1-\tau)=-\beta(\tau)$. Hence $\{0.5\}$ is
blind and any $\tau\neq0.5$ detects it.

(iii) Blind spot. For a finite $\mathcal T$, the alternatives with
$\beta(\tau_i)=0$ for all $i$ but $\beta\not\equiv0$ are exactly those missed; such
$\beta$ exist (perturb the conditional CDF only between grid levels).
\end{proof}

\section{Data generating processes}
\label{app:dgp}

All synthetic experiments draw from structural equation models sharing a fork skeleton
$Z\to X$, $Z\to Y$; the alternatives add the edge $X\to Y$, so the queried null
$X\indep Y\given Z$ holds exactly under the fork and is violated under the triangle. The
conditioning set is $Z=(Z_1,\dots,Z_{|S|})$ with independent coordinates, each
$Z_j\sim\mathcal N(0,1)$ in the Gaussian regime, $Z_j\sim t_3$ in the heavy-tailed regime,
or, in the mixed regime, $Z_1$ a three-level categorical and the rest Gaussian. Two
independent seeds fix the structure (graph, variable types, weights) and the sample, so
every cell is exactly reproducible.

\paragraph{The conditional tail edge (Section~\ref{sec:tailedge}, centerpiece).} With
$\varepsilon_X,\varepsilon_Y$ independent of $Z$,
\begin{align*}
Z_j &\sim \mathcal N(0,1),\ \ j=1,\dots,|S|; \qquad
X = \sum_{j=1}^{|S|} Z_j + \varepsilon_X, \quad \varepsilon_X\sim\mathcal N(0,1);\\
Y &= \sum_{j=1}^{|S|} Z_j + \varepsilon_Y, \qquad
\varepsilon_Y = \frac{W-\mu(\alpha)}{\sigma(\alpha)}, \quad W\sim\mathrm{SN}(\alpha), \quad \alpha = c\,X,
\end{align*}
where $W\sim\mathrm{SN}(\alpha)$ is a skew-normal of shape $\alpha$, standardized affinely by
\[
\delta=\frac{\alpha}{\sqrt{1+\alpha^2}},\qquad
\mu(\alpha)=\delta\sqrt{2/\pi},\qquad
\sigma(\alpha)^2=1-\tfrac{2}{\pi}\,\delta^2,
\]
applied per observation (at each row's own $\alpha=cX$), so that
$\E[\varepsilon_Y]=0$ and $\Var[\varepsilon_Y]=1$ for every $\alpha$. The conditional moments
of $Y$ are therefore
\[
\E[Y\mid X,Z]=\sum_{j} Z_j\ \ (\text{free of }X), \qquad \Var[Y\mid X,Z]=1,
\]
and $X$ enters $Y$ only through the conditional skew, which is monotone in
$\alpha=cX$. The strength sweep is $c\in\{1,2,3,4\}$ (Figure~\ref{fig:tail-strength}); the matched
null sets $c=0$ ($\varepsilon_Y\sim\mathcal N(0,1)$, no $X\to Y$ edge). These standardization
formulas are exactly those of the generator, so the equations above reproduce the reported
cells; we further confirmed numerically that $\E[\varepsilon_Y\mid X]$ and
$\Var[\varepsilon_Y\mid X]$ are flat in $X$ while $\mathrm{Skew}[\varepsilon_Y\mid X]$ is
monotone, matching this construction.

\paragraph{Calibration nulls (Section~\ref{sec:nuisance}).} Each null is the fork alone (no
$X\to Y$ edge), so $X\indep Y\given Z$ holds exactly. $X$ and $Y$ are drawn from the same child
specification with independent noises $\varepsilon_X,\varepsilon_Y$ (so $X$ and $Y$ are
conditionally independent by construction), and $Z=(Z_1,\dots,Z_{|S|})$ has independent
coordinates. The four mechanisms differ only in which nuisance they stress:
\begin{itemize}
\item \texttt{heavy\_tail}: $Z_j\sim t_3$; $X=\sum_j Z_j+\varepsilon_X$, $Y=\sum_j Z_j+\varepsilon_Y$,
with $\varepsilon_X,\varepsilon_Y\sim\mathcal N(0,1)$. Only the conditioners are heavy-tailed; the
noise stays Gaussian.
\item \texttt{hetero}: $Z_j\sim\mathcal N(0,1)$; linear mean $\sum_j Z_j$ on both children with a
shared $Z_1$-driven conditional scale, $\varepsilon_X=\exp(0.4\,Z_1)\,\eta_X$ and
$\varepsilon_Y=\exp(0.4\,Z_1)\,\eta_Y$, where $\eta_X,\eta_Y\sim\mathcal N(0,1)$ (the scale
argument is clipped to $|Z_1|\le5$ purely as an overflow guard, inactive with probability
$\approx1$ under the Gaussian $Z_1$).
\item \texttt{nonlin\_mean}: $Z_j\sim\mathcal N(0,1)$; $X=\sin(1.5\,Z_1)+\varepsilon_X$,
$Y=\sin(1.5\,Z_1)+\varepsilon_Y$, with $\varepsilon_X,\varepsilon_Y\sim\mathcal N(0,1)$. The mean
depends on $Z_1$ alone (additive in a single parent, which an additive spline fits but a fixed
linear basis cannot); $Z_2,\dots,Z_{|S|}$ are inert conditioners.
\item \texttt{mixed\_Z}: $Z_1$ a three-level categorical (codes $\{0,1,2\}$, uniform) and
$Z_j\sim\mathcal N(0,1)$ for $j\ge2$; $X=g(Z_1)+0.8\sum_{j\ge2} Z_j+\varepsilon_X$ and likewise for
$Y$, with stratum means $g(0),g(1),g(2)=-1,0,1$ and $\varepsilon_X,\varepsilon_Y\sim\mathcal N(0,1)$.
\end{itemize}
In every cell $|S|\in\{1,3,5\}$ sets only the number of conditioners; the stressed mechanism is
otherwise unchanged.

\paragraph{Power ladder (Section~\ref{sec:ablation}, Table~\ref{tab:power-n2k}).} Unlike the
calibration nulls, the power alternatives use nonlinear $Z$-dependence, so the conditional test
must genuinely regress out $Z$. With $Z=(Z_1,\dots,Z_{|S|})$, $Z_j\sim\mathcal N(0,1)$,
\[
X=\sin\!\Big(\textstyle\sum_j Z_j\Big)+\varepsilon_X,\quad \varepsilon_X\sim\mathcal N(0,0.7^2);
\qquad Y=0.7\textstyle\sum_j Z_j^2+\phi(X)+\varepsilon_Y .
\]
The null takes $\phi\equiv0$ and $\varepsilon_Y\sim\mathcal N(0,1)$, so $X\indep Y\given Z$; each
alternative adds an $X\to Y$ edge touching a single moment of $Y$:
\begin{itemize}
\item \texttt{linear\_mean}: $\phi(X)=0.6\,X$, $\varepsilon_Y\sim\mathcal N(0,1)$ (conditional mean).
\item \texttt{nonmonotone\_z2}: $\phi(X)=0.5\,X^2$, $\varepsilon_Y\sim\mathcal N(0,1)$ (non-monotone
mean, with $\Cov(X,Y\given Z)\approx0$).
\item \texttt{scale}: $\phi\equiv0$, $\varepsilon_Y$ heteroscedastic with conditional standard
deviation $0.3+\exp(0.5\,X)$ ($X$ clipped to $[-4,4]$): the mean is flat in $X$, the variance is
driven by $X$.
\item \texttt{tail\_shape}: $\phi\equiv0$, $\varepsilon_Y$ a unit variance skew-normal of shape
$\alpha=4X$: mean and variance fixed, only the conditional skew depends on $X$.
\end{itemize}
Each isolates the moment a given block of the feature set is built to detect; $|S|\in\{1,3,5\}$
sets the number of conditioners.

\paragraph{Mixed-type spread alternatives (Section~\ref{sec:mixed-exp}, Table~\ref{tab:mixed}).}
The conditioning set is $Z=(Z_c,Z_k)$ with $Z_c\sim\mathcal N(0,1)$ and $Z_k$ a three-level
categorical (codes $\{0,1,2\}$, uniform). The fork $Z\to X$, $Z\to Y$ is the null; each alternative
adds $X\to Y$. Continuous nodes take a stratum mean $g(Z_k)\in\{-1,0,1\}$ plus $0.8\,Z_c$ with
$\mathcal N(0,1)$ noise; categorical nodes are logistic in their parents. The four $(X,Y)$ type
pairings (C continuous, K categorical) are:
\begin{itemize}
\item \textbf{CC}: continuous $X$ drives $Y$'s conditional variance (std $0.3+\exp(0.5\,X)$, $X$
clipped to $[-4,4]$) with $Y$'s mean free of $X$: a spread edge beyond covariance.
\item \textbf{KC}: categorical $X$ (code $x$) sets $Y$'s conditional standard deviation $0.5+x$,
mean free of $X$: a categorical spread edge.
\item \textbf{CK}: continuous $X$ shifts $Y$'s class log-odds (weights $\{-2,0,2\}$ over the three
levels): a probability shift.
\item \textbf{KK}: categorical $X$ shifts $Y$'s class log-odds ($2.5$ on the matching level): a
categorical probability shift.
\end{itemize}
CC and KC are the spread edges the covariance family misses; CK and KK are the
probability shifts every calibrated test detects.

\section{Additional experimental results}
\label{app:experiments}

\paragraph{Asymmetric versus symmetric quantile product.} The quantile
block uses the asymmetric products $e_X\,r_\tau(Y)$ and $r_\tau(X)\,e_Y$ (a level
times a quantile indicator), which is what creates the F1 orientation risk
(Section~\ref{sec:fixes}); a symmetric quantile-quantile product $r_\tau(X)\,r_\tau(Y)$
would avoid F1, so the asymmetric choice requires justification
(Table~\ref{tab:asym}). The two products target different alternatives. On the skew edge that preserves mean and variance, the quantile block's reason for existing ($X$'s
level driving $Y$'s skew), the asymmetric product has power $0.90$ against the symmetric
product's $0.15$, because the symmetric product sees $X$ only through its tail indicator and
discards the mid-range level information that drives $Y$'s shape. The symmetric product
instead detects pure tail co-movement (symmetric joint extremes of $X$ and $Y$, here
a mean-independent $t$-copula), where the asymmetric product has power $0.03$. Neither
dominates. The asymmetric construction targets ``$X$'s level drives $Y$'s distribution'',
the directed, level-driven alternative that constraint-based discovery cares about; F1 and
the low power against symmetric tail co-movement are the cost of that targeting.

\begin{table}[ht]\centering
\begin{tabular}{lcc}
\toprule
alternative ($X\!\to\!Y$) & asymmetric $e_X r_\tau(Y)$ & symmetric $r_\tau(X) r_\tau(Y)$ \\
\midrule
null & 0.05 & 0.06 \\
scale & 0.99 & 1.00 \\
skew & \textbf{0.90} & 0.15 \\
tail co-movement ($t$-copula) & 0.03 & \textbf{0.99} \\
\bottomrule
\end{tabular}
\caption{Asymmetric versus symmetric quantile product ($n=2000$, $200$ reps,
$\alpha=0.05$; Wilson CIs in text). The asymmetric mean-quantile product GFCM uses
dominates on the level-driven skew edge ($0.90$ vs $0.15$); the symmetric product
instead catches symmetric tail co-movement ($0.99$ vs $0.03$). They target different
alternatives, and GFCM's choice is the one aligned with directed, level-driven dependence,
at the cost of the F1 risk and blindness to symmetric tail co-movement.}
\label{tab:asym}
\end{table}

\paragraph{Rank-transformed conditioning: robustness to heavy-tailed conditioners.} GFCM
fits its nuisance with a spline whose knots are spaced uniformly over the range of the
conditioning set. On a heavy-tailed conditioner a handful of extreme values stretch that
range, so the knots concentrate in the empty tails and leave the bulk underfit, so the
conditional mean is poorly estimated and its residual leaks into the statistic. Rank-transforming $Z$
before the spline maps each conditioner to a uniform scale, so the knots always land where
the data are, whatever the marginal shape. Table~\ref{tab:rankz} isolates the effect on a true
null $X\indep Y\given Z$ with a smooth shared mean: with rank-$Z$ the test holds nominal level
across Gaussian, Student-$t$, and log-normal conditioners, while on raw $Z$ it is calibrated only
when $Z$ is light-tailed and loses control as the tails heavy, reaching size $0.17$ at $t_2$ and
$1.00$ at $t_1$ (Cauchy). Power is untouched (an added edge is recovered at $1.00$ either way),
and a monotone but light-tailed warp (log-normal) does not trigger it, so the driver is heavy
tails, not reparametrization. Rank-$Z$ is thus a design choice for the heavy-tailed, mixed-type
regime GFCM targets, at no cost when $Z$ is well behaved; it is also the scale on which the
validity theory is stated (Appendix~\ref{app:proofs}).

\begin{table}[ht]\centering
\small
\begin{tabular}{lcc}
\toprule
conditioner $Z$ \ (null $X\indep Y\given Z$) & rank-$Z$ (default) & raw $Z$ \\
\midrule
Gaussian (control) & 0.070~[0.051--0.096] & 0.070~[0.051--0.096] \\
Student-$t_5$ & 0.090~[0.068--0.118] & 0.068~[0.049--0.094] \\
Student-$t_2$ (heavy) & 0.058~[0.041--0.082] & \textbf{0.170~[0.140--0.205]} \\
Student-$t_1$ (Cauchy) & 0.062~[0.044--0.087] & \textbf{0.998~[0.989--1.000]} \\
log-normal (monotone warp) & 0.068~[0.049--0.094] & 0.064~[0.046--0.089] \\
\midrule
\textit{power}: $t_2$ with an $X\!\to\!Y$ edge & 1.000~[0.992--1.000] & 1.000~[0.992--1.000] \\
\bottomrule
\end{tabular}
\caption{Rank-transformed conditioning holds calibration on heavy-tailed conditioners; raw $Z$
does not. GFCM size on a true null $X\indep Y\given Z$ whose smooth shared mean
$m(Z)=\tanh(1.2\,Z)$ drives both $X$ and $Y$ with independent noise ($n=2000$, $500$ reps,
nominal $0.05$; Wilson $95\%$ CIs in brackets). With the default rank-$Z$ the test stays at level
whatever the conditioner's tails; on raw $Z$ the uniform knot spline starves the bulk and size
climbs to $0.17$ ($t_2$) and $1.00$ ($t_1$). The final row confirms power is unaffected. A
light-tailed monotone warp (log-normal) does not trigger the failure, so the driver is heavy
tails, not reparametrization.}
\label{tab:rankz}
\end{table}

\paragraph{Additive versus interaction nuisance.} To isolate the scope
of assumption~(A1), Table~\ref{tab:interaction} runs the additive spline backbone alone (without the
default's interaction block) on two nulls $X\indep Y\given Z$ with $Z\in\mathbb{R}^2$: an additive
shared nuisance ($\E[X\given Z]=\sin Z_1+\sin Z_2$, $\E[Y\given Z]=\cos Z_1+\cos Z_2$) and an
interaction shared nuisance ($\E[X\given Z]=\E[Y\given Z]=Z_1Z_2$). On the additive null the backbone
holds level across $n$; on the interaction null it leaves the unfittable $Z_1Z_2$ term in both
residuals, so $\E[e_X e_Y\given Z]=(Z_1Z_2)^2>0$, a first-order bias that does not vanish, and it rejects
at $1.00$ at every $n$; adding knots does not lower it. The default design removes this term: its
degree-3 interaction block includes the $Z_1Z_2$ term and removes the bias at
any $n$ (Table~\ref{tab:nuisance-tradeoff}). Capturing interactions to order $q$ needs
${\sim}\binom{d}{q}K^q$ columns; the full tensor product costs $K^d$, infeasible at regression cost.

\begin{table}[ht]\centering
\begin{tabular}{lccc}
\toprule
shared nuisance $\E[X\given Z]=\E[Y\given Z]$ & $n=10^3$ & $n=4{\cdot}10^3$ & $n=1.6{\cdot}10^4$ \\
\midrule
additive ($\sin Z_1+\sin Z_2$ / $\cos Z_1+\cos Z_2$) & 0.06 & 0.05 & 0.05 \\
interaction ($Z_1 Z_2$) & \textbf{1.00} & \textbf{1.00} & \textbf{1.00} \\
\bottomrule
\end{tabular}
\caption{Null rejection rate (nominal $0.05$, $Z\in\mathbb{R}^2$, $300$ reps; $X\indep Y\given Z$
in both rows). The additive spline backbone alone (without the default interaction block) holds
level under an additive shared mean but collapses to $1.00$ under a pure interaction mean it cannot
represent, at every $n$: a first-order bias, not a vanishing one. This pins the structural scope of
(A1); the default design's interaction block  of degree 3 removes it (Table~\ref{tab:nuisance-tradeoff}).}
\label{tab:interaction}
\end{table}

\paragraph{Complementary polynomial block and the single index.}
The default design carries both a fixed polynomial of degree 3 interaction block and, gated
to $n\ge5000$, a cross-fitted single-index spline. Neither subsumes the other, because they fail
on opposite alternatives (Table~\ref{tab:nuisance-tradeoff}). The polynomial block includes the
interaction terms exactly, so it removes a bilinear confounder $Z_1Z_2$ at any sample size
($0.03$ at $n{=}2000$), but it cannot represent a transcendental mean such as $\sin(\sum_j Z_j)$
and its size rises to $1.00$ at every $n$. The single index instead estimates a direction
$\hat a$ and splines along it, so it captures the transcendental mean ($0.06$ once $n$ is large),
but the estimated direction is noisy at small $n$, where it underfits the very interaction the
polynomial block handles exactly ($0.20$ at $n{=}2000$, falling to $0.09$ by $n{=}10^4$). The two
cover disjoint regimes, which is why the default test keeps the polynomial block (for
interactions, at all $n$) and adds the single index gated to large $n$ (for transcendental means,
where the direction is estimable). A pure single-index nuisance, which would drop the
polynomial block to save its $O(d^2)$ cost, is ruled out by the small-$n$ interaction row, and
no significance threshold on the index recovers it (the liberality is estimation variance, not
spurious admission). The conditioning set cap and the gate on sample size are thus a regime switch
between an exact but limited nuisance and an estimated but flexible one, not free parameters.

\begin{table}[t]\centering\small
\begin{tabular}{llcc}
\toprule
null & $n$ & polynomial block & single index \\
\midrule
additive control & 2000 & 0.02~[0.01--0.05] & 0.10~[0.06--0.14] \\
interaction $Z_1Z_2$ & 2000 & \textbf{0.03~[0.01--0.06]} & 0.20~[0.15--0.26] \\
interaction $Z_1Z_2$ & 10000 & 0.05~[0.03--0.09] & 0.09~[0.05--0.13] \\
transcendental $\sin(\textstyle\sum_j Z_j)$ & 2000 & \textbf{1.00~[0.98--1.00]} & 0.15~[0.11--0.21] \\
transcendental $\sin(\textstyle\sum_j Z_j)$ & 10000 & \textbf{1.00~[0.98--1.00]} & 0.06~[0.03--0.10] \\
\bottomrule
\end{tabular}
\caption{Complementary failure of the two nuisance blocks (size, nominal $0.05$; $200$ reps,
Wilson 95\% CI; $Z\in\mathbb{R}^3$). The polynomial block removes polynomial interactions at any
$n$ but diverges on a transcendental mean; the single index represents the transcendental mean
(cleanly once $n$ is large enough to estimate its direction) but underfits the interaction at
small $n$. The default test uses both: the polynomial block for interactions, the single index
gated to large $n$ for transcendental means.}
\label{tab:nuisance-tradeoff}
\end{table}

\paragraph{Convergence with sample size: full grid.} Figure~\ref{fig:conv} plots two representative
slices (the heavy-tailed and mixed-type nulls, and the \texttt{alpha\_3} skew edge). The
complete convergence grid is presented below: every null (Table~\ref{tab:conv-null}) and alternative
(Table~\ref{tab:conv-power}) across depth $|S|\in{1,3,5}$, each row tracing a test's size or size-corrected power across the full sample-size sweep ($n$ from $500$ to $10^5$). The pattern is uniform: GFCM and BLITZ hold nominal size across every null while FFCI, the partial copula test, and boosted GCM diverge as $n$ grows.

\footnotesize\setlength{\tabcolsep}{3pt}
\begin{longtable}{l*{14}{c}}
\caption{Convergence of null size with sample size (full sweep; columns are $n$ in thousands; nominal $0.05$). PartCopula and boosted GCM are capped at moderate $n$ and shown as \texttt{--} beyond it.}\label{tab:conv-null}\\
\toprule
 & \multicolumn{14}{c}{sample size $n$ (thousands)}\\
\cmidrule(lr){2-15}
test & 0.5 & 1 & 1.5 & 2 & 2.5 & 3 & 4 & 5 & 7 & 10 & 15 & 20 & 50 & 100 \\
\midrule
\endfirsthead
\multicolumn{15}{c}{\tablename\ \thetable\ (continued)}\\
\toprule
test & 0.5 & 1 & 1.5 & 2 & 2.5 & 3 & 4 & 5 & 7 & 10 & 15 & 20 & 50 & 100 \\
\midrule
\endhead
\bottomrule
\endfoot
\midrule \multicolumn{15}{l}{\texttt{heavy\_tail}, $|S|=1$}\\
GFCM & 0.05 & 0.04 & 0.06 & 0.05 & 0.08 & 0.05 & 0.04 & 0.06 & 0.06 & 0.05 & 0.05 & 0.07 & 0.05 & 0.06 \\
BLITZ & 0.06 & 0.05 & 0.05 & 0.05 & 0.08 & 0.07 & 0.04 & 0.05 & 0.05 & 0.06 & 0.06 & 0.06 & 0.05 & 0.08 \\
FFCI & 0.06 & 0.07 & 0.11 & 0.13 & 0.20 & 0.29 & 0.37 & 0.56 & 0.79 & 0.96 & 1.00 & 1.00 & 1.00 & 1.00 \\
RCoT & 0.05 & 0.05 & 0.04 & 0.05 & 0.06 & 0.05 & 0.03 & 0.04 & 0.06 & 0.05 & 0.05 & 0.07 & 0.04 & 0.06 \\
GCM-boost & 0.18 & 0.17 & 0.21 & 0.27 & 0.27 & 0.35 & 0.45 & 0.54 & 0.64 & 0.72 & 0.80 & 0.82 & -- & -- \\
PartCopula & 0.05 & 0.10 & 0.13 & 0.17 & 0.24 & 0.30 & 0.40 & 0.54 & 0.72 & 0.85 & -- & -- & -- & -- \\
\midrule \multicolumn{15}{l}{\texttt{heavy\_tail}, $|S|=3$}\\
GFCM & 0.08 & 0.08 & 0.07 & 0.07 & 0.07 & 0.06 & 0.07 & 0.06 & 0.05 & 0.06 & 0.04 & 0.07 & 0.05 & 0.07 \\
BLITZ & 0.05 & 0.04 & 0.07 & 0.07 & 0.05 & 0.06 & 0.04 & 0.04 & 0.07 & 0.06 & 0.04 & 0.05 & 0.04 & 0.06 \\
FFCI & 1.00 & 1.00 & 1.00 & 1.00 & 1.00 & 1.00 & 1.00 & 1.00 & 1.00 & 1.00 & 1.00 & 1.00 & 1.00 & 1.00 \\
RCoT & 0.07 & 0.05 & 0.07 & 0.12 & 0.17 & 0.18 & 0.23 & 0.34 & 0.39 & 0.49 & 0.63 & 0.69 & 0.92 & 0.98 \\
GCM-boost & 0.16 & 0.10 & 0.11 & 0.12 & 0.11 & 0.10 & 0.11 & 0.17 & 0.24 & 0.61 & 0.93 & 1.00 & -- & -- \\
PartCopula & 0.75 & 1.00 & 1.00 & 1.00 & 1.00 & 1.00 & 1.00 & 1.00 & 1.00 & 1.00 & -- & -- & -- & -- \\
\midrule \multicolumn{15}{l}{\texttt{heavy\_tail}, $|S|=5$}\\
GFCM & 0.11 & 0.08 & 0.06 & 0.06 & 0.09 & 0.05 & 0.06 & 0.06 & 0.07 & 0.07 & 0.06 & 0.06 & 0.05 & 0.06 \\
BLITZ & 0.04 & 0.04 & 0.05 & 0.05 & 0.07 & 0.04 & 0.04 & 0.05 & 0.05 & 0.07 & 0.04 & 0.04 & 0.05 & 0.05 \\
FFCI & 1.00 & 1.00 & 1.00 & 1.00 & 1.00 & 1.00 & 1.00 & 1.00 & 1.00 & 1.00 & 1.00 & 1.00 & 1.00 & 1.00 \\
RCoT & 0.44 & 0.74 & 0.85 & 0.93 & 0.96 & 0.99 & 0.99 & 0.99 & 1.00 & 1.00 & 1.00 & 1.00 & 0.99 & 1.00 \\
GCM-boost & 0.18 & 0.18 & 0.22 & 0.25 & 0.27 & 0.29 & 0.35 & 0.47 & 0.70 & 0.92 & 1.00 & 1.00 & -- & -- \\
PartCopula & 0.95 & 1.00 & 1.00 & 1.00 & 1.00 & 1.00 & 1.00 & 1.00 & 1.00 & 1.00 & -- & -- & -- & -- \\
\midrule \multicolumn{15}{l}{\texttt{hetero}, $|S|=1$}\\
GFCM & 0.06 & 0.05 & 0.08 & 0.06 & 0.06 & 0.05 & 0.05 & 0.05 & 0.05 & 0.05 & 0.04 & 0.05 & 0.06 & 0.05 \\
BLITZ & 0.06 & 0.05 & 0.06 & 0.04 & 0.04 & 0.06 & 0.04 & 0.07 & 0.04 & 0.06 & 0.05 & 0.06 & 0.05 & 0.04 \\
FFCI & 0.07 & 0.03 & 0.05 & 0.03 & 0.05 & 0.05 & 0.04 & 0.06 & 0.04 & 0.07 & 0.04 & 0.05 & 0.08 & 0.09 \\
RCoT & 0.06 & 0.05 & 0.05 & 0.04 & 0.06 & 0.04 & 0.05 & 0.06 & 0.05 & 0.04 & 0.04 & 0.05 & 0.05 & 0.05 \\
GCM-boost & 0.18 & 0.09 & 0.08 & 0.06 & 0.06 & 0.04 & 0.06 & 0.05 & 0.06 & 0.06 & 0.05 & 0.03 & -- & -- \\
PartCopula & 0.05 & 0.05 & 0.05 & 0.05 & 0.06 & 0.05 & 0.06 & 0.06 & 0.05 & 0.05 & -- & -- & -- & -- \\
\midrule \multicolumn{15}{l}{\texttt{hetero}, $|S|=3$}\\
GFCM & 0.09 & 0.07 & 0.07 & 0.06 & 0.09 & 0.08 & 0.05 & 0.06 & 0.07 & 0.07 & 0.06 & 0.06 & 0.05 & 0.05 \\
BLITZ & 0.06 & 0.05 & 0.06 & 0.05 & 0.07 & 0.07 & 0.04 & 0.05 & 0.05 & 0.06 & 0.05 & 0.05 & 0.05 & 0.07 \\
FFCI & 0.59 & 0.94 & 0.99 & 1.00 & 1.00 & 1.00 & 1.00 & 1.00 & 1.00 & 1.00 & 1.00 & 1.00 & 1.00 & 1.00 \\
RCoT & 0.07 & 0.05 & 0.06 & 0.06 & 0.08 & 0.07 & 0.06 & 0.07 & 0.06 & 0.07 & 0.06 & 0.07 & 0.07 & 0.10 \\
GCM-boost & 0.14 & 0.09 & 0.08 & 0.10 & 0.09 & 0.09 & 0.08 & 0.07 & 0.07 & 0.08 & 0.04 & 0.06 & -- & -- \\
PartCopula & 0.28 & 0.78 & 0.93 & 0.99 & 1.00 & 1.00 & 1.00 & 1.00 & 1.00 & 1.00 & -- & -- & -- & -- \\
\midrule \multicolumn{15}{l}{\texttt{hetero}, $|S|=5$}\\
GFCM & 0.17 & 0.12 & 0.10 & 0.09 & 0.06 & 0.08 & 0.09 & 0.08 & 0.06 & 0.07 & 0.07 & 0.07 & 0.06 & 0.05 \\
BLITZ & 0.05 & 0.07 & 0.05 & 0.06 & 0.06 & 0.05 & 0.06 & 0.04 & 0.05 & 0.06 & 0.04 & 0.05 & 0.05 & 0.05 \\
FFCI & 1.00 & 1.00 & 1.00 & 1.00 & 1.00 & 1.00 & 1.00 & 1.00 & 1.00 & 1.00 & 1.00 & 1.00 & 1.00 & 1.00 \\
RCoT & 0.15 & 0.22 & 0.34 & 0.35 & 0.38 & 0.46 & 0.48 & 0.50 & 0.57 & 0.58 & 0.68 & 0.73 & 0.82 & 0.84 \\
GCM-boost & 0.14 & 0.13 & 0.11 & 0.11 & 0.11 & 0.09 & 0.08 & 0.11 & 0.10 & 0.10 & 0.09 & 0.12 & -- & -- \\
PartCopula & 0.38 & 0.93 & 0.99 & 1.00 & 1.00 & 1.00 & 1.00 & 1.00 & 1.00 & 1.00 & -- & -- & -- & -- \\
\midrule \multicolumn{15}{l}{\texttt{mixed\_Z}, $|S|=1$}\\
GFCM & 0.05 & 0.04 & 0.06 & 0.04 & 0.06 & 0.05 & 0.05 & 0.06 & 0.05 & 0.06 & 0.05 & 0.07 & 0.07 & 0.05 \\
BLITZ & 0.06 & 0.06 & 0.07 & 0.04 & 0.07 & 0.05 & 0.06 & 0.07 & 0.05 & 0.07 & 0.04 & 0.05 & 0.06 & 0.05 \\
FFCI & 0.52 & 0.98 & 1.00 & 1.00 & 1.00 & 1.00 & 1.00 & 1.00 & 1.00 & 1.00 & 1.00 & 1.00 & 1.00 & 1.00 \\
RCoT & 0.05 & 0.06 & 0.06 & 0.04 & 0.07 & 0.05 & 0.05 & 0.06 & 0.05 & 0.05 & 0.05 & 0.03 & 0.05 & 0.05 \\
GCM-boost & 0.06 & 0.05 & 0.05 & 0.04 & 0.06 & 0.04 & 0.05 & 0.06 & 0.05 & 0.05 & 0.04 & 0.04 & -- & -- \\
PartCopula & -- & -- & -- & -- & -- & -- & -- & -- & -- & -- & -- & -- & -- & -- \\
\midrule \multicolumn{15}{l}{\texttt{mixed\_Z}, $|S|=3$}\\
GFCM & 0.06 & 0.06 & 0.06 & 0.04 & 0.07 & 0.07 & 0.05 & 0.04 & 0.06 & 0.05 & 0.04 & 0.05 & 0.05 & 0.07 \\
BLITZ & 0.05 & 0.04 & 0.04 & 0.06 & 0.08 & 0.06 & 0.04 & 0.05 & 0.06 & 0.05 & 0.04 & 0.04 & 0.05 & 0.04 \\
FFCI & 1.00 & 1.00 & 1.00 & 1.00 & 1.00 & 1.00 & 1.00 & 1.00 & 1.00 & 1.00 & 1.00 & 1.00 & 1.00 & 1.00 \\
RCoT & 0.05 & 0.05 & 0.06 & 0.05 & 0.06 & 0.08 & 0.06 & 0.05 & 0.07 & 0.05 & 0.04 & 0.06 & 0.05 & 0.04 \\
GCM-boost & 0.15 & 0.10 & 0.13 & 0.11 & 0.08 & 0.10 & 0.08 & 0.07 & 0.06 & 0.06 & 0.05 & 0.06 & -- & -- \\
PartCopula & -- & -- & -- & -- & -- & -- & -- & -- & -- & -- & -- & -- & -- & -- \\
\midrule \multicolumn{15}{l}{\texttt{mixed\_Z}, $|S|=5$}\\
GFCM & 0.08 & 0.08 & 0.08 & 0.06 & 0.08 & 0.04 & 0.06 & 0.06 & 0.05 & 0.06 & 0.06 & 0.06 & 0.05 & 0.06 \\
BLITZ & 0.06 & 0.04 & 0.06 & 0.05 & 0.04 & 0.05 & 0.05 & 0.06 & 0.04 & 0.04 & 0.06 & 0.05 & 0.04 & 0.04 \\
FFCI & 1.00 & 1.00 & 1.00 & 1.00 & 1.00 & 1.00 & 1.00 & 1.00 & 1.00 & 1.00 & 1.00 & 1.00 & 1.00 & 1.00 \\
RCoT & 0.08 & 0.10 & 0.13 & 0.19 & 0.20 & 0.27 & 0.28 & 0.33 & 0.38 & 0.44 & 0.56 & 0.56 & 0.70 & 0.76 \\
GCM-boost & 0.14 & 0.10 & 0.11 & 0.11 & 0.08 & 0.08 & 0.08 & 0.07 & 0.08 & 0.07 & 0.08 & 0.09 & -- & -- \\
PartCopula & -- & -- & -- & -- & -- & -- & -- & -- & -- & -- & -- & -- & -- & -- \\
\midrule \multicolumn{15}{l}{\texttt{nonlin\_mean}, $|S|=1$}\\
GFCM & 0.09 & 0.05 & 0.07 & 0.06 & 0.07 & 0.03 & 0.04 & 0.04 & 0.04 & 0.05 & 0.06 & 0.06 & 0.06 & 0.05 \\
BLITZ & 0.06 & 0.04 & 0.07 & 0.05 & 0.07 & 0.06 & 0.05 & 0.06 & 0.05 & 0.07 & 0.06 & 0.06 & 0.06 & 0.04 \\
FFCI & 0.06 & 0.04 & 0.06 & 0.05 & 0.07 & 0.06 & 0.06 & 0.06 & 0.06 & 0.08 & 0.09 & 0.10 & 0.25 & 0.43 \\
RCoT & 0.07 & 0.06 & 0.05 & 0.05 & 0.07 & 0.06 & 0.04 & 0.04 & 0.04 & 0.06 & 0.05 & 0.04 & 0.05 & 0.04 \\
GCM-boost & 0.18 & 0.09 & 0.08 & 0.05 & 0.09 & 0.05 & 0.06 & 0.05 & 0.05 & 0.06 & 0.04 & 0.03 & -- & -- \\
PartCopula & 0.05 & 0.03 & 0.05 & 0.07 & 0.06 & 0.07 & 0.10 & 0.10 & 0.10 & 0.15 & -- & -- & -- & -- \\
\midrule \multicolumn{15}{l}{\texttt{nonlin\_mean}, $|S|=3$}\\
GFCM & 0.10 & 0.10 & 0.10 & 0.10 & 0.11 & 0.11 & 0.14 & 0.05 & 0.06 & 0.07 & 0.06 & 0.07 & 0.04 & 0.07 \\
BLITZ & 0.09 & 0.14 & 0.14 & 0.11 & 0.11 & 0.11 & 0.08 & 0.06 & 0.05 & 0.05 & 0.04 & 0.07 & 0.06 & 0.04 \\
FFCI & 0.05 & 0.05 & 0.06 & 0.06 & 0.06 & 0.08 & 0.08 & 0.10 & 0.10 & 0.16 & 0.20 & 0.28 & 0.58 & 0.89 \\
RCoT & 0.08 & 0.08 & 0.07 & 0.08 & 0.08 & 0.07 & 0.07 & 0.04 & 0.07 & 0.07 & 0.05 & 0.06 & 0.05 & 0.08 \\
GCM-boost & 0.16 & 0.10 & 0.07 & 0.09 & 0.08 & 0.10 & 0.08 & 0.07 & 0.08 & 0.07 & 0.06 & 0.06 & -- & -- \\
PartCopula & 0.02 & 0.06 & 0.06 & 0.06 & 0.05 & 0.05 & 0.05 & 0.08 & 0.08 & 0.08 & -- & -- & -- & -- \\
\midrule \multicolumn{15}{l}{\texttt{nonlin\_mean}, $|S|=5$}\\
GFCM & 0.10 & 0.06 & 0.09 & 0.07 & 0.07 & 0.08 & 0.06 & 0.07 & 0.08 & 0.05 & 0.07 & 0.06 & 0.05 & 0.06 \\
BLITZ & 0.05 & 0.04 & 0.06 & 0.06 & 0.04 & 0.06 & 0.06 & 0.07 & 0.04 & 0.05 & 0.04 & 0.03 & 0.04 & 0.03 \\
FFCI & 0.04 & 0.05 & 0.05 & 0.05 & 0.06 & 0.05 & 0.06 & 0.05 & 0.06 & 0.06 & 0.06 & 0.04 & 0.08 & 0.07 \\
RCoT & 0.08 & 0.07 & 0.07 & 0.07 & 0.06 & 0.07 & 0.07 & 0.06 & 0.08 & 0.05 & 0.06 & 0.05 & 0.05 & 0.07 \\
GCM-boost & 0.13 & 0.09 & 0.09 & 0.07 & 0.08 & 0.07 & 0.07 & 0.06 & 0.04 & 0.06 & 0.06 & 0.06 & -- & -- \\
PartCopula & 0.03 & 0.04 & 0.04 & 0.05 & 0.04 & 0.05 & 0.07 & 0.07 & 0.05 & 0.04 & -- & -- & -- & -- \\
\end{longtable}
\normalsize\setlength{\tabcolsep}{6pt}

\footnotesize\setlength{\tabcolsep}{3pt}
\begin{longtable}{l*{14}{c}}
\caption{Convergence of size-corrected power with sample size (full sweep; columns are $n$ in thousands).}\label{tab:conv-power}\\
\toprule
 & \multicolumn{14}{c}{sample size $n$ (thousands)}\\
\cmidrule(lr){2-15}
test & 0.5 & 1 & 1.5 & 2 & 2.5 & 3 & 4 & 5 & 7 & 10 & 15 & 20 & 50 & 100 \\
\midrule
\endfirsthead
\multicolumn{15}{c}{\tablename\ \thetable\ (continued)}\\
\toprule
test & 0.5 & 1 & 1.5 & 2 & 2.5 & 3 & 4 & 5 & 7 & 10 & 15 & 20 & 50 & 100 \\
\midrule
\endhead
\bottomrule
\endfoot
\midrule \multicolumn{15}{l}{\texttt{tail\_shape}, $|S|=1$}\\
GFCM & 0.32 & 0.67 & 0.86 & 0.99 & 1.00 & 1.00 & 1.00 & 1.00 & 1.00 & 1.00 & 1.00 & 1.00 & 1.00 & 1.00 \\
BLITZ & 0.09 & 0.12 & 0.13 & 0.17 & 0.19 & 0.23 & 0.26 & 0.32 & 0.50 & 0.57 & 0.76 & 0.88 & 1.00 & 1.00 \\
FFCI & 0.30 & 0.59 & 0.78 & 0.92 & 0.94 & 0.98 & 1.00 & 1.00 & 1.00 & 1.00 & 1.00 & 1.00 & 1.00 & 1.00 \\
RCoT & 0.14 & 0.23 & 0.30 & 0.41 & 0.42 & 0.50 & 0.58 & 0.67 & 0.73 & 0.79 & 0.88 & 0.89 & 0.96 & 0.98 \\
GCM-boost & 0.04 & 0.04 & 0.04 & 0.08 & 0.05 & 0.05 & 0.06 & 0.06 & 0.06 & 0.05 & 0.06 & 0.09 & -- & -- \\
PartCopula & 0.08 & 0.06 & 0.06 & 0.06 & 0.02 & 0.05 & 0.07 & 0.04 & 0.06 & 0.05 & -- & -- & -- & -- \\
\midrule \multicolumn{15}{l}{\texttt{tail\_shape}, $|S|=3$}\\
GFCM & 0.11 & 0.40 & 0.71 & 0.90 & 0.98 & 0.99 & 1.00 & 1.00 & 1.00 & 1.00 & 1.00 & 1.00 & 1.00 & 1.00 \\
BLITZ & 0.07 & 0.12 & 0.13 & 0.13 & 0.18 & 0.21 & 0.25 & 0.34 & 0.41 & 0.50 & 0.70 & 0.81 & 1.00 & 1.00 \\
FFCI & 0.11 & 0.16 & 0.16 & 0.21 & 0.29 & 0.35 & 0.51 & 0.63 & 0.76 & 0.91 & 1.00 & 0.99 & 1.00 & 1.00 \\
RCoT & 0.05 & 0.09 & 0.10 & 0.07 & 0.11 & 0.12 & 0.14 & 0.18 & 0.24 & 0.28 & 0.39 & 0.31 & 0.51 & 0.52 \\
GCM-boost & 0.06 & 0.07 & 0.07 & 0.06 & 0.05 & 0.06 & 0.06 & 0.05 & 0.07 & 0.06 & 0.07 & 0.04 & -- & -- \\
PartCopula & 0.04 & 0.04 & 0.06 & 0.05 & 0.03 & 0.03 & 0.03 & 0.03 & 0.05 & 0.03 & -- & -- & -- & -- \\
\midrule \multicolumn{15}{l}{\texttt{tail\_shape}, $|S|=5$}\\
GFCM & 0.07 & 0.25 & 0.67 & 0.86 & 0.98 & 1.00 & 1.00 & 1.00 & 1.00 & 1.00 & 1.00 & 1.00 & 1.00 & 1.00 \\
BLITZ & 0.13 & 0.16 & 0.17 & 0.21 & 0.27 & 0.29 & 0.31 & 0.45 & 0.54 & 0.69 & 0.85 & 0.92 & 1.00 & 1.00 \\
FFCI & 0.06 & 0.09 & 0.09 & 0.15 & 0.15 & 0.17 & 0.22 & 0.35 & 0.43 & 0.55 & 0.73 & 0.80 & 1.00 & 1.00 \\
RCoT & 0.05 & 0.11 & 0.08 & 0.07 & 0.06 & 0.11 & 0.09 & 0.12 & 0.09 & 0.12 & 0.13 & 0.12 & 0.18 & 0.26 \\
GCM-boost & 0.08 & 0.05 & 0.07 & 0.03 & 0.07 & 0.05 & 0.07 & 0.07 & 0.05 & 0.06 & 0.04 & 0.05 & -- & -- \\
PartCopula & 0.04 & 0.07 & 0.11 & 0.07 & 0.05 & 0.06 & 0.08 & 0.09 & 0.08 & 0.11 & -- & -- & -- & -- \\
\midrule \multicolumn{15}{l}{\texttt{alpha\_3}, $|S|=1$}\\
GFCM & 0.36 & 0.70 & 0.92 & 0.99 & 0.99 & 1.00 & 1.00 & 1.00 & 1.00 & 1.00 & 1.00 & 1.00 & 1.00 & 1.00 \\
BLITZ & 0.07 & 0.11 & 0.13 & 0.19 & 0.19 & 0.24 & 0.30 & 0.36 & 0.50 & 0.58 & 0.78 & 0.88 & 1.00 & 1.00 \\
FFCI & 0.14 & 0.34 & 0.44 & 0.70 & 0.81 & 0.88 & 0.98 & 1.00 & 1.00 & 1.00 & 1.00 & 1.00 & 1.00 & 1.00 \\
RCoT & 0.09 & 0.14 & 0.22 & 0.28 & 0.35 & 0.44 & 0.52 & 0.62 & 0.72 & 0.77 & 0.86 & 0.87 & 0.99 & 0.98 \\
GCM-boost & 0.04 & 0.05 & 0.04 & 0.08 & 0.05 & 0.06 & 0.05 & 0.07 & 0.05 & 0.04 & 0.06 & 0.08 & -- & -- \\
PartCopula & 0.04 & 0.10 & 0.06 & 0.05 & 0.05 & 0.04 & 0.05 & 0.05 & 0.08 & 0.08 & -- & -- & -- & -- \\
\midrule \multicolumn{15}{l}{\texttt{alpha\_3}, $|S|=3$}\\
GFCM & 0.10 & 0.29 & 0.52 & 0.68 & 0.87 & 0.92 & 1.00 & 1.00 & 1.00 & 1.00 & 1.00 & 1.00 & 1.00 & 1.00 \\
BLITZ & 0.07 & 0.09 & 0.10 & 0.11 & 0.11 & 0.13 & 0.18 & 0.26 & 0.28 & 0.30 & 0.53 & 0.70 & 0.97 & 1.00 \\
FFCI & 0.00 & 0.00 & 0.00 & 0.00 & 0.00 & 0.00 & 0.00 & 0.00 & 0.00 & 0.00 & 0.00 & 0.00 & 0.00 & 0.00 \\
RCoT & 0.08 & 0.11 & 0.11 & 0.08 & 0.07 & 0.12 & 0.21 & 0.23 & 0.30 & 0.41 & 0.44 & 0.53 & 0.74 & 0.85 \\
GCM-boost & 0.03 & 0.08 & 0.07 & 0.03 & 0.05 & 0.03 & 0.05 & 0.03 & 0.06 & 0.06 & 0.06 & 0.06 & -- & -- \\
PartCopula & 0.08 & 0.07 & 0.07 & 0.12 & 0.17 & 0.14 & 0.12 & 0.16 & 0.19 & 0.28 & -- & -- & -- & -- \\
\midrule \multicolumn{15}{l}{\texttt{alpha\_3}, $|S|=5$}\\
GFCM & 0.05 & 0.14 & 0.24 & 0.41 & 0.60 & 0.74 & 0.90 & 0.93 & 0.99 & 1.00 & 1.00 & 1.00 & 1.00 & 1.00 \\
BLITZ & 0.06 & 0.06 & 0.07 & 0.09 & 0.12 & 0.13 & 0.17 & 0.15 & 0.24 & 0.28 & 0.42 & 0.52 & 0.92 & 1.00 \\
FFCI & 0.00 & 0.00 & 0.00 & 0.00 & 0.00 & 0.00 & 0.00 & 0.00 & 0.00 & 0.00 & 0.00 & 0.00 & 0.00 & 0.00 \\
RCoT & 0.03 & 0.04 & 0.02 & 0.02 & 0.02 & 0.03 & 0.02 & 0.04 & 0.00 & 0.00 & 0.00 & 0.00 & 0.00 & 0.00 \\
GCM-boost & 0.06 & 0.03 & 0.09 & 0.05 & 0.06 & 0.03 & 0.06 & 0.05 & 0.05 & 0.03 & 0.06 & 0.03 & -- & -- \\
PartCopula & 0.04 & 0.07 & 0.07 & 0.06 & 0.08 & 0.10 & 0.09 & 0.16 & 0.11 & 0.16 & -- & -- & -- & -- \\
\end{longtable}
\normalsize\setlength{\tabcolsep}{6pt}

\end{document}